\documentclass[lettersize,journal]{IEEEtran}
\usepackage{amsmath}
\usepackage{amsfonts}
\usepackage{algorithmic}
\usepackage{array}
\usepackage{textcomp}
\usepackage{stfloats}
\usepackage{url}
\usepackage{verbatim}
\usepackage{graphicx}
\usepackage{subfigure}
\usepackage{amssymb} 
\usepackage{cases}%分段函数
\usepackage{algorithm}
\usepackage{cite}
\usepackage{booktabs} % For formal tables
\usepackage{multirow} %multirow for format of table 
\usepackage{xcolor}
\usepackage{bm}%定理

\newtheorem{theorem}{\textbf{Theorem}}%定理
\newtheorem{lemma}{\textbf{Lemma}}%引理
\newenvironment{proof}{{\it Proof.}\quad}{\hfill $\blacksquare$\par}
\def\BibTeX{{\rm B\kern-.05em{\sc i\kern-.025em b}\kern-.08em
    T\kern-.1667em\lower.7ex\hbox{E}\kern-.125emX}}
\let\fff= \scriptscriptstyle
\usepackage{balance}
\begin{document}
%\title{Enhancing Covert Communication With Multi-antenna in Two-hop Relay Systems}
%\title{Exploiting Multi-antenna for Covert Communication in Two-hop Relay Systems}
%\title{Exploiting Multi-antenna for Covert Communication Enhancement in Relay Systems}
\title{Enhancing Covert Communication in Relay Systems Using Multi-Antenna Technique}
\author{He Zhu, Huihui Wu, Wei Su, Xiaohong Jiang
%\\School of Systems Information Science, Future University Hakodate, 116-2\\Kamedanakano-cho, Hakodate, Hokkaido, 041-8655,  Japan \IEEEmembership{Member, IEEE},

\thanks{H. Zhu and X. Jiang are with the School of System Information Science, Future University Hakodate, Hakodate 041-8655, Japan (emails: zhuhe0203@foxmail.com; jiang@fun.ac.jp).}
\thanks{Huihui Wu is with Beijing National Research Center for Information Science and Technology (BNRist), Department of Automation, Tsinghua University, Beijing 100084, China (email: hhwu1994@mail.tsinghua.edu.cn).}
\thanks{Wei Su is with the School of Electronic and Information Engineering, Beijing Jiaotong University, Beijing 100044, China (email: wsu@bjtu.edu.cn).}
}
%\markboth{Journal of \LaTeX\ Class Files,~Vol.~18, No.~9, September~2020}%
%{How to Use the IEEEtran \LaTeX \ Templates}

\maketitle

\begin{abstract}

This paper exploits the multi-antenna technique to enhance the covert communication performance in a relay system, where a source S conducts covert communication with a destination D via a relay R, subjecting to the detections of transmissions in the two hops from a single-antenna warden W. To demonstrate the performance gain from adopting the multi-antenna technique, we first consider the scenario when S, R and D all adopt single antenna, and apply hypothesis testing and statistics theories to develop a theoretical framework for the covert performance modeling in terms of detection error probability (DEP) and covert throughput. We then consider the scenario when S, R and D all adopt multiple antennas, and apply the hypothesis testing, statistics and matrix theories to develop corresponding theoretical framework for performance modeling. We further explore the optimal designs of the target rate and transmit power for covert throughput maximization under above both scenarios, subjecting to the constraints of covertness, reliability and transmit power. To solve the optimization problems, we employ Karushi-Kuhn-Tucker (KKT) conditions method in the single antenna scenario and a search algorithm in the multi-antenna scenario. Finally, we provide extensive numerical results to illustrate how the multi-antenna technique can enhance the covert performance in two-hop relay systems.

%A classic direct search algorithm is also proposed to solve these optimization problems. 

% and also develop alternative search algorithm to solve the corresponding optimization problems. Finally, we provide extensive numerical results to demonstrate that adopting the multi-antenna technique can significantly enhance the covert performance in the two-hop relay system.

\end{abstract}

\begin{IEEEkeywords}
Covert communication, Multi-antenna technique, two-hop relay system, physical layer security.
\end{IEEEkeywords}

\section{Introduction}
Covert communication, also known as low probability detection (LPD) communication, employs a variety of signal processing technologies (e.g., coding, modulation, waveform design, etc.) to achieve the secure transmission between legitimate users under a very low probability of being detected by a warden \cite{makhdoom2022comprehensive}\cite{fu2019modulation}. The covert communication has found many critical applications in government, military, unmanned aerial vehicle (UAV), Internet of Things networks, etc. The recent studies indicate that the two-hop relay systems are highly appealing for implementing covert communication \cite{wu2020covert,hu2017covert,hu2018covert,forouzesh2020covert,gao2021covert,su2020covert,sun2021covert}, since such systems can largely alleviate the constraints of transmit power, channel fading and path loss in conventional single-hop covert wireless communication system \cite{wang2020covert,hu2019covert,shahzad2020covert,ta2019covert}, and thus can significantly improve the performance and expend the coverage of covert communication.\par
% which can be applied in many critical situations, such as covert government and military operations, hide the location of unmanned aerial vehicle (UAV), avoid cybercrime in the Internet of Things (IoT) networks. By now, extensive works mainly focus on the basic signal-hop covert wireless communication, where a transmitter-receiver pair aims to conduct direct communication and hide the existence of transmission itself against the warden\cite{wang2020covert,he2018covert,hu2019covert,yang2019achieving,bash2016covert,shahzad2020covert,yan2018delay,he2017covert,ta2019covert}. Due to the low power transmission adopted in covert wireless communication, channel fading and path loss significantly reduce the covert communication performance, and thus the covert transmission distance limitation. \par
%The cooperative relay has been widely adopted for achieving long distance transmission and expanding the coverage of covert wireless communication. As a first step, extensive research efforts have been devoted to the study of two-hop covert wireless communications. Seminal work in \cite{wu2020covert} proved that in AWGN channels, the maximum number of covert bits of the system is $\textit{O} \sqrt{n} $, which is independent of the relaying patterns and prior knowledge. It is proved that two-hop relay systems also satisfy the square root law (SRL) in covert wireless communication. Although the SRL gives the bits that can be transmitted without being detected by the warden under the covert transmission, the covert channel is a zero rate channel. \par
By now, considerable research efforts have been dedicated to investigating the study of covert wireless communications in two-hop relay systems. The seminal work in \cite{wu2020covert} provides a fundamental study for covert wireless communication in two-hop relay systems and proves that covert rate in such systems also follows the square root law (SRL). The works in \cite{hu2017covert, hu2018covert} consider a relay system where the greedy relay uses the power allocated from source to transmit addition covert signal to the destination subject to the detection from the source, and respectively proposed rate and power control-based covert schemes. In \cite{forouzesh2020covert}, the authors consider an untrusted relay systems and reach the maximization covert performance. The works in \cite{gao2021covert,su2020covert} explore the covert communication through relay selection scheme of multiple relay systems. The work in \cite{sun2021covert} investigates the two-hop relay covert communication system where the relay in both full-duplex (FD) and half-duplex (HD) modes, and design a joint FD/HD mode where the relay can switch flexibly between these two modes to maximize the covert performance. Some recent works also explore the covert communications with the help of mobile UAV relays \cite{jiang2021covert,zhang2022uav} or with the assist of relays equipped with intelligent reflecting surface (IRS) \cite{chen2021enhancing}.\par

The above works help us understand the potential applications of two-hop relay systems in covert communications. It is notable that these works mainly focus on the simple scenario where each node adopts single antenna in covert wireless communication. However, some recent studies show that the multi-antenna technique can significantly improve the covert communication performance in single-hop systems in terms of detection error probability (DEP), covert rate, covert throughput, etc \cite{zheng2019multi,yang2019achieving,chen2021multi,lin2021multi}. Thus, a natural question is that is it possible to further enhance the covert performance by resorting to multi-antenna technique in the two-hop relay systems. To address this problem, this paper, for the first time, provides solid theoretical frameworks to investigate the impact of adopting multi-antenna technique on covert communication performance in a two-hop relay system.\par
The main contributions of this paper are summarized as follows.
\noindent{}
\begin{itemize}
\item{We investigate the covert communication in a two-hop relay system, where a source S conducts covert communication with a destination D via a decode and forward (DF) relay R, subjecting to the detections of transmissions in the two hops from a single-antenna warden W. To demonstrate the performance gain from adopting the multi-antenna technique, we consider two scenarios where S, R and D all adopt single antenna and multiple antennas, respectively.}
%We first consider a two-hop covert transmission with decode and forward (DF) multi-antenna relay assistance with channel and noise uncertainty. We conduct the theoretical modeling for detection error probability of the two-hop relay system under the single antenna scheme. We then formulate the covert throughput under the constraints of covertness, reliability and transmit power.}
\item{Under the single-antenna scenario, we apply the hypothesis testing and statistics theories to develop a theoretical framework for the covert performance modeling in terms of DEP and covert throughput. We further apply the hypothesis testing, statistics and matrix theories to develop a theoretical framework for corresponding performance modeling under the multi-antenna scenario.}
%Using the same assumption, we also model the multi-antenna covert transmission in which source, relay and destination are all equipped with multiple antennas. The transmitter employs transmission antenna selection (TAS), and the receiver employs the maximum ratio combination (MRC) algorithm in order to optimally feedback the output signal of the transmitter to the receiver. To analyze the performance of covert communication within multi-antenna two-hop relay network, we derive the detection error probability of warden and then determine the covert throughput subject to the power constraint, reliability constraint, and covert constraint to represent the system's reliability.} 
\item{We explore the optimal designs of the target rate and transmit power for covert throughput maximization under above both scenarios, subjecting to covertness, reliability and transmit power constraints. To solve the corresponding optimization problems, we employ KKT conditions method in the single antenna scenario and develop a classic search algorithm in the multi-antenna scenario. }
%We further formulate an optimization design problem under these two schemes. Through the joint use of alternative optimized target rate and transmit power for covert throughput maximization under the aforementioned two schemes, and also develop alternative search algorithms to solve the corresponding optimization problems.}
\item{We provide extensive numerical results to illustrate the potential impact and enhancement of adopting multi-antenna technique on covert communication performance in two-hop relay systems.}
\end{itemize}

The remainder of the paper is organized as follows. Section II presents the system model and preliminaries of transmission schemes and detection method for warden. Section III develops theoretical models on DEP and covert throughput under both single-antenna and multi-antenna scenarios. The covert throughput maximization problems and related alternative search algorithm are discussed in Section IV, and comprehensive numerical results are provided in Section V. Finally, Section VI provides the conclude for this work.\par
%The optimal covert throughput is conducted and alternative search algorithms of the two models are proposed in Section IV. Some numerical results are shown in Section V. Finally the conclusion is made in Section VI. \par
\textit{Notations:} Scalars are presented as lower-case letters. Vector are presented by bold lower-case letters, while matrices are presented by bold upper-case letters. For a given vector, $(\cdot)^{H}$ represents the operation of conjugate transpose, $(\cdot)^{T}$ denotes the transpose operation, and ${ \parallel \cdot \parallel }$ represents the norm of the vector. $\mathbb E(\cdot)$ denotes expectation operator.

\section{System Model and Preliminaries}
As shown in Fig.1, we consider a two-hop decode-and-forward (DF) relay system, where a source $S$ sends covert message to a destination $D$ via a relay $R$ in the presence of an illegal detection warden $W$. It is assumed that all nodes adopt the HD mode and there is no direct transmission link between $S$ and $D$ due to channel fading and obstructions such as large buildings, so the transmission from $S$ to $D$ requires two time slots (i.e., $S \rightarrow R$ and $R \rightarrow D$ transmissions). In the first time slot, $S$ transmits signal to $R$ while $D$ remains silent. In the second time slot, $R$ decodes and forwards the received signal to $D$. We use the quasi-static Rayleigh fading channel to model the channels, where the channel fading coefficients are constant within a single time slot, but change independently and randomly from time slot to time slot.\par
For the considered two-hop DF relay system, we consider both the single-antenna and the multi-antenna scenario depicted in Fig. 1(a) and Fig. 1(b), respectively. In the single-antenna scenario, all nodes are assumed to be equipped with one omnidirectional antenna. The coefficient of the channel from $I$ to $J$ is denoted as $h_{\fff IJ}$, where $IJ \in \{SR, RD, SW, RW\}$. In the multi-antenna scenario, $W$ is assumed to be equipped with a single antenna, while $S$ and $D$ are assumed to be equipped with $N_{\fff S}$ and $N_{\fff D}$ antennas, $R$ is equipped with $N_{\fff R_{r}}$ receiving antennas and $N_{\fff R_{t}}$ transmitting antennas, respectively. The channel matrixes of $S$-$R$, $R$-$D$, $S$-$W$ and $R$-$W$ are denoted as $ \mathbf{H}_{\fff SR} \in  \mathbb{C}^{N_{\fff R_{r}} \times N_{\fff S}}$, $\mathbf{H}_{\fff RD} \in  \mathbb{C}^{N_{\fff D} \times N_{\fff R_{t}}}$, $ \mathbf{H}_{\fff SW} \in  \mathbb{C}^{1 \times N_{\fff S}}$ and $ \mathbf{H}_{\fff RW} \in  \mathbb{C}^{1 \times N_{\fff R_{t}}}$, respectively. %Same as previous work \cite{sun2020resource}, we assume that in both scenarios, the channel coefficients are available for $R$ and $D$, but are unknown for $W$.\par

%We assume that in both two scenarios, the channel coefficients are available for $R$ and $D$, but are unknown for $W$. These assumptions have been widely used in previous research and can be practically achieved by sending pilot symbols from $S$ and $R$ at the beginning of each time slot to help the receivers (i.e. $R$, $D$) estimate the channel state information (CSI) of the transmission channel \cite{sun2020resource}.\par

%Alice encodes the covert information into $N$ symbols and transmits them $N$ times during each communication slot.
%\begin{figure}[h]
%\centering 
%\includegraphics[width=0.5\textwidth]{fig0.jpg}
%\includegraphics[width=8cm]{fig0.png}
%\caption{System model} 
%\end{figure}
\begin{figure}[h]
    \centering
    \subfigure[Single-antenna scenario]{
        \includegraphics[width=0.5\textwidth]{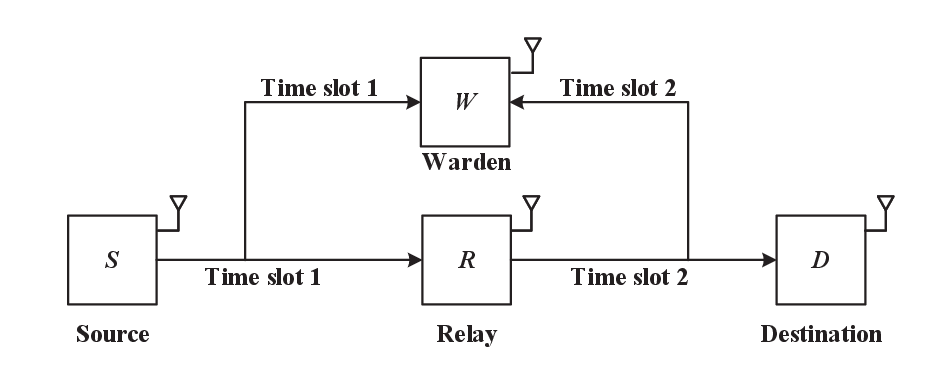}
    }
      \quad    %用 \quad 来换行
    \subfigure[Multi-antenna scenario]{
	\includegraphics[width=0.5\textwidth]{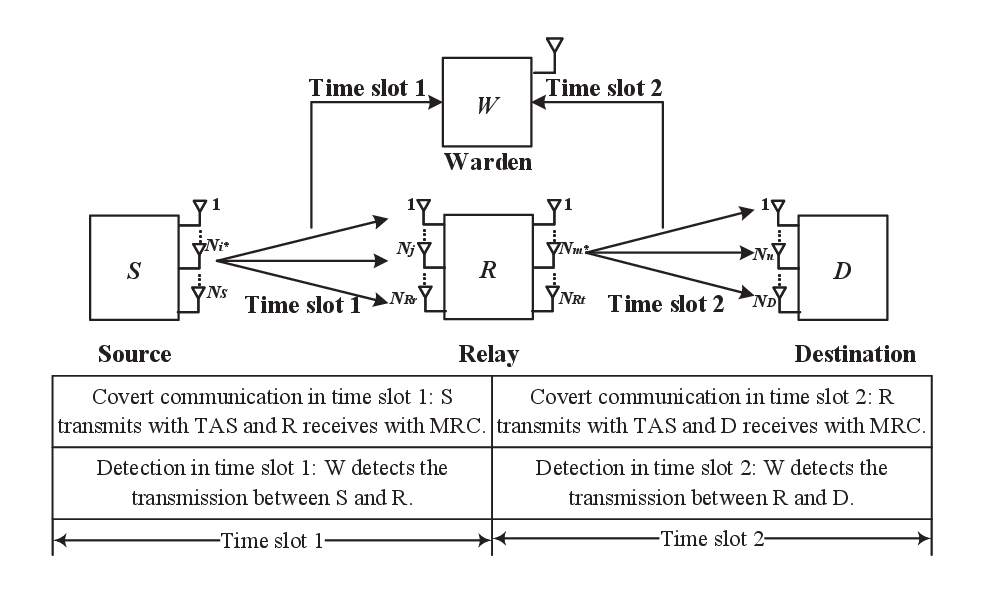}
    }
    \caption{Two-hop relay systems.}
    \label{fig.0}
\end{figure}

\subsection{Transmission Scheme in Single-antenna Scenario}

\textbf{Transmission in the first time slot:} Considering that $S$ encodes the signal into $N$ symbols and transmits them in the first time slot. Let $x_{\fff S}(k)$ denote the $k$-th symbol transmitted from $S$ to $R$, here $x_{\fff S}(k)$ satisfies the unit power constraint $\mathbb E[\mid{x_{\fff S}(k)}\mid^2] = 1$ and $k = 1,2, \cdots, N$. Thus, for the symbol $x_{\fff S}(k)$, the corresponding signal $ y_{\fff R}^{(1)}(k)$ received at $R$ can be determined as
% when $N$ is large enough (i.e., $N \rightarrow \infty$),
%Considering $S$ encodes the covert information into $N$ symbols and transmits them in each communication slot, in this work we assume that $N$ is large enough (i.e., $N \rightarrow \infty$). Thus, the received signal at $R$ is given by
 \begin{equation}
\begin{aligned}
y_{\fff R}^{(1)}(k)=  \sqrt{P_{\fff S}}{h}_{\fff SR} x_{\fff S}(k) + n_{\fff SR}(k),\label{s-SR}
\end{aligned}
\end{equation}
where $P_{\fff S}$ is the transmit power of $S$ subject to transmit power constraint $ P_{max}$; $n_{\fff SR}(k)$ is the additive white Gaussian noise (AWGN) at $ R$ with zero mean and variance $\sigma _{\fff SR}^2$, i.e., $n_{\fff SR}(k)\sim \mathcal{CN}({0}, \sigma _{\fff SR}^2)$.

\textbf{Transmission in the second time slot:} After receiving the signal $y_{\fff R}^{(1)}(k)$, $R$ then decodes and forwards it to $D$. Let $x_{\fff R}(k)$  denote the $k$-th symbol transmitted from $R$ to $D$, here $\mathbb E[\mid{x_{\fff R}(k)}\mid^2] = 1$ and $k= 1,2,\cdots, N$. For the symbol $x_{\fff R}(k)$, the corresponding signal $y_{\fff D}^{(2)}(k)$ received at $D$ can be determined as
 %Then, $R$ receives the signals from $S$, and  Let $x_{\fff R}(k)$  denotes the $k$-th symbol with the unit power transmitted from $R$, which satisfies $\mathbb E[\mid{x_{\fff R}(k)}\mid^2] = 1$. The received signal $y_{\fff D}^{(2)}(k)$ in the second time slot for the $k$-th channel use at $D$ can be given by
 \begin{equation}
\begin{aligned}
y_{\fff D}^{(2)}(k)=  \sqrt{P_{\fff R}}{h}_{\fff RD} x_{\fff R}(k) + n_{\fff RD}(k),\label{s-RD}
\end{aligned}
\end{equation}
where $P_{\fff R}$ is the transmit power of $R$ subject to the transmit power constraint $ P_{max}$; $n_{\fff RD}(k)$ is the AWGN at $D$ satisfying $n_{\fff RD}(k)\sim \mathcal{CN}({0}, \sigma _{\fff RD}^2)$, and $ \sigma _{\fff RD}^2$ is the noise variance of the receiving antenna at $D$.

\subsection{Transmission Scheme in Multi-antenna Scenario}
%We consider some typical transmission schemes in each time slot of multi-antenna scenario. In the first time slot, $S$ adopts transmit antenna selection (TAS) scheme to send signals and $R$ receives signals with maximal-ratio combining (MRC). In the second time slot, $R$ transmits signals based on TAS scheme and $D$ employs MRC scheme to receive the signals\cite{chen2005analysis}. We note that the TAS scheme has a bit-level feedback overhead. This is due to the fact that the receivers only need to feed back the index of the strongest antenna (i.e., the transmit antenna providing the largest instantaneous SNR at receiver). Thus, S and R only know the statistical characterization of channel coefficients.\par
For the multi-antenna scenario, we consider the typical transmit antenna selection (TAS) scheme for signal transmission and the maximal-ratio combining (MRC) scheme for receiving signal \cite{chen2005analysis}.\par
%transmission schemes proposed in \cite{chen2005analysis}. In the first time slot, $S$ adopts the transmit antenna selection (TAS) scheme to send signals while $R$ adopts the maximal-ratio combining (MRC) scheme to receive signals. In the second time slot, $R$ transmits signals based on the TAS scheme and $D$ employs the MRC scheme to receive signals.\par
\textbf{Transmission in the first time slot:} %Under the TAS scheme, the best single transmit antenna is selected at $S$ to maximize the instantaneous signal-to-noise ratio (SNR) $\gamma _{\fff N_{i}R}$ at $R$. Then the 
 The  channel matrix $\mathbf{H}_{\fff  {SR}}$ of $S\rightarrow R$ transmission is a $N_{\fff R_{r}} \times N_{\fff S}$ matrix defined as 
\begin{equation}
\begin{aligned}
\mathbf{H}_{\fff  {SR}}=
 \begin{bmatrix}
h_{\fff 1,1} & \cdots & h_{\fff 1,N_{\fff S}}\\
\cdots & \cdots & \cdots \\ 
h_{\fff N_{\fff R_{r}},1} & \cdots & h_{\fff N_{\fff R_{r}},N_{\fff S}}
 \end{bmatrix},
\end{aligned}
\end{equation}
where $h_{ji}$ is the fading coefficient from the $i$-th antenna at $S$ to the $j$-th antenna at $R$, $i\in \{1,...N_{\fff S} \}$, $j\in \{1,...N_{\fff R_{r}} \}$. Under the TAS transmission scheme, the antenna with the highest instantaneous signal-to-noise ratio (SNR) will be selected from $N_{\fff S}$ antennas at $S$ to transmit signals from $S$ to $R$. We denote the index of the selected antenna at $S$ as $i^{*}$, and denote the channel vector of the antenna $i^*$ for $S \rightarrow R$ transmission as $ \mathbf{h}_{\fff i^{*}R}=(h_{\fff 1,i^{*}}, h_{\fff 2,i^{*}},\cdots ,  h_{\fff N_{\fff R_{r}},i^{*}})^\mathrm{T}$. We consider that the signal $x_{\fff S}=\{x_{\fff S}(1), x_{\fff S}(2),\cdots,x_{\fff S}(N)\}$ transmitted from $S$ to $D$ is consisted of $N$ symbols, where symbol $x_{\fff S}(k)$ satisfies the condition of $\mathbb E[|x_{\fff S}(k) |^2]=1$, $k=1,2,\cdots, N$. For the symbol $x_{\fff S}(k)$, the corresponding signal vector $\mathbf y_{\fff R}^{(1)}(k)$ received by $N_{\fff R_{r}}$ antennas at $R$ is given by
% $x_{\fff S}$ transmitted from $S$ to $D$ is divided into a signal sequence $x_{\fff S}=\{x_{\fff S}(1), x_{\fff S}(2),\cdots,x_{\fff S}(N)\}$. $x_{\fff S}(k)$ represents the $k$-th symbol and satisfies $\mathbb E[|x_{\fff S}(k) |^2]=1$ with $k=1,2,\cdots, N$. Thus, the symbol vector $\mathbf Y_{\fff R}^{(1)}(k)$ received by $N_{\fff R}$ antennas at $R$ for the $k$-th channel use is given by
% and $( \cdot )^\mathrm{H}$ denotes conjugate transpose operators.
%$ \mathbf{h}_{\fff N_{i}R}=(h_{\fff N_{1}N_{i}}, h_{\fff N_{2}N_{i}},\cdots ,  h_{\fff N_{\fff R}N_{ i}})^\mathrm{T}$
%where $R$ and $S$ exchange the best signal-to-noise ratio (SNR) $ \gamma _{\fff N_{i}R} $ of the $i$-th antenna of $R$. The channel matrix of the selected antenna for $S \rightarrow R$ transmissions can be represented by $ \mathbf{h}_{\fff N_{i}R}=(h_{\fff N_{1}N_{i}}, h_{\fff N_{2}N_{i}},\cdots ,  h_{\fff N_{\fff R}N_{ i}})^\mathrm{T}$, where vector $\mathbf{h}_{\fff N_{i} R}$ is a column of $\mathbf{H}_{\fff SR}$, $ i \in  \{1,\cdots ,N_{\fff S}\}\ $ and $( \cdot )^\mathrm{T}$ denotes the transpose operation. Thus, the signal received at $R$ for each symbol period is given by
\begin{equation}
\begin{aligned}
%\mathbf y_{\fff R}=  \sqrt{P_{\fff S}}\mathbf{h}_{\fff N_{i}R}^{H}\mathbf{w}_1 x +\mathbf n_{\fff SR},
\mathbf y_{\fff R}^{(1)}(k)=  \sqrt{P_{\fff S}}\mathbf{h}_{\fff i^{*} R} x_{\fff S}(k) +\mathbf n_{\fff SR},\label{v_yr}
\end{aligned}
\end{equation}
where $P_{\fff S}$ is the transmit power of the selected antenna $i^{*}$; $\mathbf n_{\fff SR}$ is the $N_{\fff R_{r}}\times 1$ AWGN vector at $R$ satisfying $\mathbb E[\mathbf n_{\fff SR}\mathbf n_{\fff SR}^H]=\mathbf I_{N_{\fff R_{r}}}\sigma_{\fff {SR}}^2$, and $\sigma_{\fff SR}^2$ is the noise variance of a receiving antenna.

Based on $\mathbf y_{\fff R}^{(1)}(k)$, the relay $R$ then applies the MRC scheme to receive the signal of symbol $x_{\fff S}(k)$. With the MRC scheme, $R$ combines the signals from all $N_{\fff R_{r}}$ antennas by weighting each signal according to its SNR to yield the maximum SNR ratio. Thus, for the symbol $x_{\fff S}(k)$, the final scalar signal ${y}_{\fff R}^{(1)}(k)$ received at $R$ can be determined as 
%Then, $R$ combines the received signals using the MRC receiving scheme, which means that the signal from each antenna of $N_{\fff R}$ is weighted according to its SNR, such that the signals from all antennas are combined to yield the maximum ratio. Therefore, MRC can convert the received symbol vector into the scalar symbol. Thus, the scalar symbol $\text{y}_{\fff R}^{(1)}(k)$ based on the symbol vector $\mathbf Y_{\fff R}^{(1)}(k)$ in (\ref{v_yr}) for the $k$-th channel use in the first time slot  at $R$ can be given by%denotes the independent zero mean circular symmetric complex Gaussian random vectors for $R$, $\mathbf n_{\fff SR} \sim \mathcal{CN}(\mathbf{0}, \sigma _{\fff SR}^2{\mathbf I_{N_{\fff R}}})$
\begin{equation}
\begin{aligned}
%\mathbf y_{\fff R}=  \sqrt{P_{\fff S}}\mathbf{h}_{\fff N_{i}R}^{H}\mathbf{w}_1 x +\mathbf n_{\fff SR},$ \mathbf{h}_{\fff N_{j}D}=(h_{\fff N_{1}N_{j}}, h_{\fff N_{2}N_{j}}, \cdots, h_{\fff {N_{\fff D}}N_{j}})^\mathrm{T}$
{y}_{\fff R}^{(1)}(k)= \mathbf{w}_1\mathbf{y}_{\fff R}^{(1)}(k)=\mathbf{w}_1 \sqrt{P_{\fff S}}\mathbf{h}_{\fff i^{*} R} x_{\fff S}(k) + \mathbf{w}_1 \mathbf{n}_{\fff SR}\label{yr_mrc},
\end{aligned}
\end{equation}
where $\mathbf w_1= {\mathbf h_{\fff i^{*} R}^{H}} / { \parallel \mathbf h_{\fff i^{*} R} \parallel}$ is the weighting vector at $R$.

\textbf{Transmission in the second time slot:} After receiving the scalar signal ${y}_{\fff R}^{(1)}(k)$, $R$ then decodes and forwards it to $D$. Under the TAS scheme, we denote the index of the selected antenna from $N_{\fff R_{t}}$ antennas at $R$ as $m^{*}$, where $m^*\in \{1,\cdots, N_{\fff R_{t}}\}$, and denote the channel vector of the antenna $m^*$ for $R\rightarrow D$ transmission as $ \mathbf{h}_{\fff m^{*}D}=(h_{\fff 1,m^{*}}, h_{\fff 2,m^{*}}, \cdots, h_{\fff {N_{\fff D}},m^{*}})^\mathrm{T}$, where $h_{nm^{*}}$ is the fading coefficient from the $m^{*}$ antenna to the $n$-th antenna at $D$, $n\in \{1,...N_{\fff D} \}$. Let $x_{\fff R}(k)$ denote the $k$-th symbol transmitted from $R$ to $D$, here $\mathbb E[|x_{\fff R}(k)|^2]=1$ and $k=1,2,\cdots, N$. Again, $D$ will apply the MRC scheme to receive the signal from $R$. Following an approach similar to (\ref{v_yr}) and (\ref{yr_mrc}), we know that for the symbol $x_{\fff R}(k)$, the corresponding scalar signal ${y}_{\fff D}^{(2)}(k)$ received at $D$ can be determined as 
%Let $h_{nm}$ denote the fading coefficient from the $m$-th antenna at $R$ to the $n$-th antenna at $D$, $m\in\{1,\dots N_{\fff R_{t}}\}$, $n\in\{1,\dots N_{\fff D}\}$. 
%where $h_{nm^{*}}$ is the fading coefficient from the $j^{*}$ antenna to the $n$-th antenna at $D$, $n\in \{1,...N_{\fff D} \}$
%$R$ receives the scalar signal ${y}_{\fff R}^{(1)}(k)$ from $S$, and then decodes and forwards it to $D$. We denote the index of the selected antenna from $N_{\fff R}$ antennas using the TAS scheme at $R$ as $m^{*}$, and denote the channel vector of the antenna $m^*$ for $R\rightarrow D$ transmission as $ \mathbf{h}_{\fff m^{*}D}=(h_{\fff 1,m^{*}}, h_{\fff 2,m^{*}}, \cdots, h_{\fff {N_{\fff D}},m^{*}})^\mathrm{T}$, where $h_{nm^{*}}$ is the fading coefficient from the $m^{*}$ antenna to the $n$-th antenna at $D$, $n\in \{1,...N_{\fff D} \}$. Let $x_{\fff R}(k)$ denote the $k$-th symbol transmitted from $R$ to $D$, which satisfies $\mathbb E[|x_{\fff R}(k)|^2]=1$. Under the MRC scheme, by following the approach similar to (\ref{v_yr}) and (\ref{yr_mrc}), for the symbol $x_{\fff R}(k)$, the scalar signal ${y}_{\fff D}^{(2)}(k)$ received at $D$ can be determined as

\begin{equation}
\begin{aligned}
{y}_{\fff D}^{(2)}(k)=\mathbf w_2 \sqrt{P_{\fff R}}\mathbf h_{\fff m^{*}D} x_{\fff R}(k) + \mathbf{w}_2 \mathbf{n}_{\fff RD},\label{yd_mrc}
\end{aligned}
\end{equation}
%where $\mathbf w_2= {\mathbf h_{\fff j^{*}D}^{H}}/{ \parallel \mathbf h_{\fff j^{*}D} \parallel } $ is the weighting vector at $D$; $P_{\fff R}$ is the transmit power of the selected antenna $m^{*}$; $\mathbf n_{\fff RD}$ is the $N_{\fff D}\times 1$ AWGN vector at $D$ satisfying $\mathbf n_{\fff RD} \sim \mathcal{CN}(\mathbf{0}, \sigma _{\fff RD}^2{\mathbf I_{N_{\fff D}}})$, and $\sigma_{\fff RD}^2$ is the noise variance of a receiving antenna. \par
where $\mathbf w_2= {\mathbf h_{\fff m^{*}D}^{H}}/{ \parallel \mathbf h_{\fff m^{*}D} \parallel } $ is the weighting vector at $D$; $P_{\fff R}$ is the transmit power of the selected antenna $m^{*}$; $\mathbf n_{\fff RD}$ is the $N_{\fff D}\times 1$ AWGN vector at $D$ satisfying $\mathbf n_{\fff RD} \sim \mathcal{CN}(\mathbf{0}, \sigma _{\fff RD}^2{\mathbf I_{N_{\fff D}}})$, and $\sigma_{\fff RD}^2$ is the noise variance at $D$. \par

\subsection{Detection Scheme at Warden}
%This work considers a detection scheme at warden, in which $W$ adopts the received energy detection method in each time slot to detect whether covert transmission exists between $S$ and $D$ \cite{shahzad2017covert}. 
Same as \cite{shahzad2017covert}, this work considers that $W$ adopts the energy detection method in each time slot to detect the covert transmission between $S$ and $D$. More specifically, in each time slot, $W$ adopts a hypothesis test based on his received signals to make a binary decision, where the null hypothesis $H_0$ means that there is no communication and the alternative hypothesis $H_1$ indicates that the communication exists.\par

\textbf{Detection in the first time slot:} To detect the existence of the covert transmission, $W$ needs to distinguish the following hypothesis. When $S$ does not send the symbol $x_{\fff S}(k)$ to $R$, \text{$H_0$} is true; When $S$ has sent the symbol $x_{\fff S}(k)$ to $R$, \text{$H_1$} is true. Under the binary hypothesis test, for the symbol $x_{\fff S}(k)$, the corresponding  signal $\text y_{\fff W}^{(1)}(k)$ received at $W$ is given by
\begin{equation}
\begin{aligned}
\text y_{\fff W}^{(1)}(k)= 
\begin{cases}
n_{\fff W}(k), & \text{if $H_0$ is true}\\ 
\!\!\sqrt{P_{\fff S}}h_{\fff SW}x_{\fff S}(k)+n_{\fff W}(k),  & \text{if $H_1$ is true,}\label{1-w}
\end{cases}
\end{aligned}
\end{equation}
where $n_{\fff W}(k)$ is the AWGN at $W$ with zero mean and variance $\sigma_{\fff W}^2$, i.e., $ n_{\fff W}(k) \sim \mathcal{CN}(0, \sigma_{\fff W}^2)$. We consider that $W$ lacks knowledge of the exact noise power (i.e., noise uncertainty), which arises due to the factors such as the variance of environment, temperature fluctuations, and calibration error. Thus, a commonly used bounded uncertainty model is adopted, where the noise power is assumed to exist within a finite range around the nominal noise power. As described in \cite{he2017covert}, the noise power $\sigma_{\fff W}^2$ at $W$ follows a logarithmic uniform distribution, and the probability density function (PDF) of  $\sigma_{\fff W}^2$ can be determined as
 \begin{equation}
\begin{aligned}
f_{ \sigma_{\fff W}^{2}(x)}=  \begin{cases} \frac{1}{2x\ln (\rho) } , & \text{if}\  \frac{\sigma_n^2}{\rho}\leq x\leq \rho\sigma_n^2 \\0, & \text{otherwise},\end{cases}\label{pdf-nw}
\end{aligned}
\end{equation}
where $\rho>1$ is the parameter that quantifies the size of the uncertainty; $ \sigma_n^2$ is the nominal noise power.\par
%Note that the existence of noise uncertainty can make $W$'s SNR lower than the required SNR for signal detection, so that detection errors occur with a certain probability. Thus, it is possible for the system to achieve covert communication.\par 
\textbf{Detection in the second time slot:} Similarly, when $R$ does not transmit the symbol $x_{\fff R}(k)$ to $D$, \text{$H_0$} is true; When $R$ has transmitted the symbol $x_{\fff R}(k)$ to $D$, \text{$H_1$} is true. Under the binary hypothesis test, for the symbol $x_{\fff R}(k)$, the corresponding signal $\text y_{\fff W}^{(2)}(k)$ received at $W$ can be given by
\begin{equation}
\begin{aligned}
\text y_{\fff W}^{(2)}(k)= 
\begin{cases}
n_{\fff W}(k), & \text{if $H_0$ is true}\\ 
\!\!\sqrt{P_{\fff R}}h_{\fff RW}x_{\fff R}(k)+n_{\fff W}(k),  & \text{if $H_1$ is true.}\label{2-w}
\end{cases}
\end{aligned}
\end{equation}

\textbf{Decision in two-hop relay system:} Based on the observations, $W$ use a power detector (e.g., radiometer) to decide whether the transmission exists or not. Thus, the decision rule in each time slot is given by
\begin{equation}
\begin{aligned}
\overline{P}_{\fff W}^{(1,2)}\triangleq  \frac{1}{N}  \sum\limits_{k=1}^N \mid \text y_{\fff W}^{(1,2)}\mid ^2 {\mathop{\gtrless}\limits^{\mathcal{D}_1}_{\mathcal{D}_0}} \tau\label{tau},
\end{aligned}
\end{equation}
where $\overline{P}_{\fff W}^{(1,2)}$ denotes the average power of received signals at $W$; $\tau$ denotes the detection threshold; $\mathcal{D}_1$ and $\mathcal{D}_0$ are the binary decisions that imply whether the transmission exists or not, respectively. According to the Lebesgue's dominated convergence theorem, the average power of the received signal $\text y_{\fff W}$ in (\ref{1-w}) or (\ref{2-w}) can be rewritten by
\begin{equation}
\begin{aligned}
\overline{P}_{\fff W}^{(1,2)}= 
\begin{cases}
\sigma_{\fff W}^2, & \text{if $H_0$ is true}\\ 
\!\!{P_{\fff I}}\mid h_{\fff IW}\mid ^2 +\sigma_{\fff W}^2,  & \text{if $H_1$ is true,}\label{p-w}
\end{cases}
\end{aligned}
\end{equation}
where $P_{\fff I}$ is the transmit power at $S$ or $R$, and $I\in \{S,R\}$.

In general, the hypothesis test introduces two types of detection errors. One is false alarm (FA), which means that $W$ makes the decision of $\mathcal{D}_1$ while the transmission does not exist. The other is missed detection (MD), which means that $W$ makes the decision of $\mathcal{D}_0$ while the transmission exists. We denote the probabilities of FA and MD as $\mathcal{P}_{FA}$ and $\mathcal{P}_{MD}$, then the total DEP $\mathcal{P}_e$ is given by
\begin{equation}
\begin{aligned}
\mathcal{P}_e= \mathcal{P}_{FA} +  \mathcal{P}_{MD}.\label{dep}
\end{aligned}
\end{equation}
Note that the detection performance at $W$ becomes worse as $\mathcal{P}_e$ approaches to 1, so the constraint of $\mathcal{P}_e$ can be determined as
%Note that the closer $\mathcal{P}_e$ approaches to 1, the larger the detection error probability is, and thus the worse detection performance at $W$. The constraint of $\mathcal{P}_e$ is given by 
\begin{equation}
\begin{aligned}
\mathcal{P}_e   \geq  1-\epsilon,\label{epsilon}
\end{aligned}
\end{equation}
here $\epsilon$ represents the covertness constraint. For a given covertness constraint in terms of  $\epsilon >0$, a covert communication is achieved if (\ref{epsilon}) holds  \cite{bash2013limits}.
 
 To obtain the final detection result in the two-hop relay system, $W$ combines the detection result in each time slot. Thus, the total DEP $\xi$ in the two-hop relay system is  given by
 \begin{comment}
\begin{equation}
\begin{aligned}
\mathcal{D}_{\{p \vee q\}} = \mathcal{D}_p^{(1)} \vee \mathcal{D}_q^{(2)}, \label{decision}
\end{aligned}
\end{equation}
where $p, q  \in \{0,1\}$.

 which is shown in Tab. 1.
\begin{table}[h]

  \caption{Detection scheme in the two-hop relay network }
  \label{tab:freq}
  \centering
  \begin{tabular}{cccc}
    \toprule
    \!\!\!\!\!Cases & \!\!\!\!\!Detection in $1^{st}$ time slot & \!\!\!\!\!Detection in $2^{nd}$ time slot&\!\!\!\!Final decision\\
    \midrule
    1 & $\mathcal{D}_0^{(1)}$ & $\mathcal{D}_0^{(2)}$ & $\mathcal{D}_0$\\
    2 & $\mathcal{D}_0^{(1)}$ & $\mathcal{D}_1^{(2)}$ & $\mathcal{D}_1$\\
    3 & $\mathcal{D}_1^{(1)}$ & $\mathcal{D}_0^{(2)}$ & $\mathcal{D}_1$\\
    4 & $\mathcal{D}_1^{(1)}$ & $\mathcal{D}_1^{(2)}$ & $\mathcal{D}_1$\\
  \bottomrule
\end{tabular}
\end{table}\par
\end{comment}

%According to the final decision in (\ref{decision}), the total detection error probability $\xi$ in this two-hop relay system can be given as
\begin{equation}
\begin{aligned}
\xi = 1-(1-\mathcal{P}_e^{(1)})(1-\mathcal{P}_e^{(2)}),\label{2-dep}
\end{aligned}
\end{equation}
where $\mathcal{P}_e^{(1)}$ and $\mathcal{P}_e^{(2)}$ are the DEP in the first and second time slots, respectively.

We note that the detection scheme at $W$ is the same in both the multi-antenna scenario and single-antenna scenario. Therefore, under the TAS scheme, $W$ cannot obtain the additional antenna gain from the selected $i^*$ antenna of $S$ or the $m^*$ antenna of $R$. For the multi-antenna scenario, we denote the fading coefficients of the $i^*$ antenna for  $S \rightarrow W$ transmission and the $m^*$ antenna for $R\rightarrow W$ transmission as $h_{\fff {i^{*}}W}$ and $h_{\fff {m^{*}}W}$, respectively. Thus, the expressions of the received signal at $W$ will be the same as (\ref{1-w}) and (\ref{2-w}) by replacing $h_{\fff SW}$ and $h_{\fff RW}$ with $h_{\fff {i^{*}}W}$ and $h_{\fff {m^{*}}W}$, respectively.

%This is because the main channel and the detection channel are independent, the TAS scheme in the multi-antenna scenario is
%In the multi-antenna scenario, the strategy for $W$ is the same as single-antenna scenario. Since the main channel and the warden's channel are independent, from $W$'s perspective, the TAS scheme is a random transmission scheme. \includegraphics[]{New_IEEEtran_how-to.pdf}As a result, $W$ is unable to benefit additional diversity from $S$. The channel gain coefficient for the $S \rightarrow W$ and $R\rightarrow W$ transmissions  denoted as $h_{\fff S_{i}W}$ and $h_{\fff R_{i}W}$, respectively, which represent the $h_{\fff SW}$ in (\ref{1-w}), (\ref{2-w}) and (\ref{p-w}).

\section{Covert Performance Analysis }
This section provides the theoretical analysis for the DEP and covert throughput of the concerned two-hop relay system under single-amtenna and multi-antenna scenarios, respectively.
%To depict the covert performance in both single-antenna and multi-antenna scenarios of the two-hop relay system, this section provides related theoretical analysis on the detection error probability and covert throughput of the system under each scenario, respectively.
%In this section, to depict the covert performance in single-antenna and multi-antenna scenarios of two-hop relay systems, we analyze detection error probability and covert throughput of two-hop relay in each scenario. %Note that the analysis of detection error probability in two models is the same due to the single antenna deployed at warden as well as the main channel and the warden's channel being independent.
\subsection{ DEP in Single-antenna and Multi-antenna Scenarios}
%Recall that the analysis of detection error probability in single-antenna and multi-antenna scenarios is the same due to warden deploying single antenna as well as the main channel and the warden's channel being independent.
As per (\ref{2-dep}), to determine DEP in the single/multiple antenna scenarios, we need to analyze $\mathcal{P}_e^{(1)}$ and $\mathcal{P}_e^{(2)}$ in the first and second time slots, respectively. We first derive $\mathcal{P}_e^{(1)}$ in the first time slot. Based on the definition of $\mathcal{P}_e$ in (\ref{dep}), we need to derive $\mathcal{P}_{FA}^{(1)}$ and $\mathcal{P}_{MD}^{(1)}$, which are given by
%As per (\ref{2-dep}), to obtain the closed-form result of detection error probability in single/multiple antenna scenarios of two-hop relay systems, we need to analyze $\mathcal{P}_e^{(1)}$ and $\mathcal{P}_e^{(2)}$ in the first and second time slots, respectively. We first derive $\mathcal{P}_e^{(1)}$ in the first time slot. Based on the definition of $\mathcal{P}_e$ as in (\ref{dep}), we derive $\mathcal{P}_{FA}^{(1)}$ and $\mathcal{P}_{MD}^{(1)}$, which are given by
\begin{equation}
\begin{aligned}
\mathcal{P}_{FA}^{(1)}=\mathbb{P}(D_1\mid H_0)=\mathbb{P}(\overline{P}_{\fff W}^{(1)}> \tau \mid H_0),\label{pfa}
\end{aligned}
\end{equation}
and
\begin{equation}
\begin{aligned}
\mathcal{P}_{MD}^{(1)}=\mathbb{P}(D_0\mid H_1)=\mathbb{P}(\overline{P}_{\fff W}^{(1)} \leq \tau \mid H_1).\label{pmd}
\end{aligned}
\end{equation}
Substituting (\ref{p-w}), (\ref{pfa}) and (\ref{pmd}) into (\ref{dep}), the DEP $\mathcal{P}_e^{(1)}$ can be determined as
\begin{equation}
\begin{aligned}
%\mathcal{P}_e &= \mathcal{P}_{FA}+ \mathcal{P}_{MD} \\
\mathcal{P}_e^{(1)} = 1- \mathbb{P}(  \tau-{P_{\fff S}}\mid h_{\fff SW}\mid ^2 \leq \sigma_{\fff W}^2 < \tau).\label{pe_1}
\end{aligned}
\end{equation}
\begin{comment}
According to the central limit theorem \cite{he2017covert}, when the number of signal samples observed at $W$ is infinite (i.e., $N \rightarrow  \infty$), the DEP $\mathcal{P}_e^{(1)}$ at $W$ in the first time slot can be expressed as
\begin{equation}
\begin{aligned}
\mathcal{P}_e^{(1)} = \begin{cases} 1, & \sigma_{\fff W}^2 \leq \tau \leq P_{\fff w}^{(1)}+\sigma_{\fff W}^2,\\0, & \text{otherwise,}\end{cases}\label{xi}
%\approx
\end{aligned}
\end{equation}
where $ P_{\fff w}^{(1)}=P_{\fff S}|h_{\fff SW}|^2$ denotes the power received signal at $W$ in the first time slot when $H_1$ is true. Then, the PDF of $P_{\fff w}^{(1)}$ is given by
\begin{equation}
\begin{aligned}
{f_{P_{\fff w}^{(1)}}}(x)= \frac{1}{{P_{\fff S}}} x \exp (- \frac{x}{P_{\fff S}}  ).\label{pdf-pw1}
\end{aligned}
\end{equation}\par
%Due to the noise uncertainty at $W$, we need to determine the detection error probability $\mathcal{P}_e^{(1)}$ at $W$ in different cases under a given $\tau$. 
\end{comment}
Let $ P_{\fff w}^{(1)}=P_{\fff S}|h_{\fff SW}|^2$ denote the power of the received signal at $W$ in the first time slot when $H_1$ is true, then the PDF of $P_{\fff w}^{(1)}$ can be determined as
\begin{equation}
\begin{aligned}
{f_{P_{\fff w}^{(1)}}}(x)= \frac{1}{{P_{\fff S}}} \exp (- \frac{x}{P_{\fff S}}  ).\label{pdf-pw1}
\end{aligned}
\end{equation}\par
Thus, according to (\ref{pdf-np}), (\ref{pe_1}) and (\ref{pdf-pw1}), $\mathcal{P}_e^{(1)}$ of $W$  in the first time slot can be given by the following lemma.\par

\begin{lemma}
For the concerned two-hop relay system with transmit power $P_{\fff S}$ for $S$, noise uncertainty $\rho$ and nominal noise power $\sigma_n^2$ of the system, for a given detection threshold $\tau$ at $W$, the DEP ${\mathcal{P}_e}^{(1)} $ in the first time slot of single/multiple antenna scenarios can be given by (\ref{pe1}), where $\mu_1=\frac{\sigma^2_n}{\rho}$, $\mu_2=\rho \sigma^2_n$, $\upsilon_1(\tau)= -\frac {\tau-\mu_1}{ P_{\fff S}}$, $ \upsilon_2(\tau) = -\frac{\tau-\mu_2}{P_{\fff S}}$. Here, $ \mathbf{Ei}( \cdot )$ is the  exponential integral function and $ \mathbf{Ei}(x)=- \int_{-x}^{\infty}  \frac{e^{-t}}{t}dt,  (x<0).$
%长公式1
\begin{figure*}[h!]
 \normalsize
% \hrulefill
	\vspace*{4pt}
 \begin{equation}
 \begin{aligned}
{\mathcal{P}_e}^{(1)}(\tau)=
\begin{cases}
1, &\text{if}\  \tau \leq\mu_1 \\
 1\!-\!\frac{1}{2\ln(\rho)}\! \left[\ln(\mu_2-\tau)+e^{-\frac{\tau}{P_{\fff S}}}\left( \mathbf{Ei}(\frac{\mu_1}{P_{\fff S}})-\mathbf{Ei}(\frac{\tau}{P_{\fff S}})\right)\ln(\frac{\tau}{\mu_1}) \right], &\text{if}\ \mu_1  < \tau \leq \mu_2 \\
1\!-\!\frac{1}{2\!\ln(\rho)}\left[ e^{-\frac{\tau}{P_{\fff S}}} \left( \mathbf{Ei}(\frac{\mu_1}{P_{\fff S}}+\mathbf{Ei}(\frac{\mu_2}{P_{\fff S}})-e^{\frac{\mu_1}{P_{\fff S}}}\ln{(\mu_1)}+e^{\frac{\mu_2}{P_{\fff S}}}\ln{(\mu_2)}\right)+1\right], &\text{if}\  \tau > \mu_2.
\end{cases} 
\label{pe1}
\end{aligned}	 
\end{equation}
	\hrulefill
	\vspace*{4pt}
\end{figure*}
\end{lemma}\par

\begin{proof}
Based on (\ref{pe_1}), the DEP $\mathcal{P}_e^{(1)}$ can be determined as
\begin{equation}
\begin{aligned}
{\mathcal{P}_e}^{(1)} &=1- \int_{0}^{\infty}\int_{0}^{\infty} f_{\sigma_{\fff W}^2}(x)dxf_{P_{\fff w}^{(1)}}(y)dy.\label{pe1-2}
%&=1-( \int_{0}^{\infty}\int_{\tau}^{\infty} f_{\sigma_{\fff W}^2}(x)dxf_{P_{\fff W}^{(1)}}(y)dy \\
%& \quad +\int_{0}^{\infty}\int_{0}^{\tau - y} f_{\sigma_{\fff W}^2}(x)dxf_{P_{\fff W}^{(1)}}(y)dy)\label{ADEP}
\end{aligned}
\end{equation}
%According to (\ref{pdf-np}), the detection error probability is influenced by noise uncertainty. Hence, the result of ${\mathcal{P}_e}^{(1)}$ in (\ref{pe1-2}) is related to the relationship between the detection threshold $\tau$ and $\mu_{1}$ as well as $\mu_2$. Thus, the result can be divided into the following three cases.\par
From (\ref{pe1-2}) we can see that ${\mathcal{P}_e}^{(1)}$ is related to the detection threshold $\tau$, $\mu_{1}$ and $\mu_2$. Thus, the analysis of (\ref{pe1-2}) can be divided into the following three cases.\par
%\emph{Case I} ${(\tau \leq \mu _{1}):}$ In this case, as shown in Fig. 2(a), when the given detection threshold is less than the minimum background noise at $W$, the FA is bound to occur at $W$ (i.e., $\mathcal{P}_{FA}^{(1)}=1$, $\mathcal{P}_{MD}^{(1)}=0$). Thus, the detection error probability is given by
\emph{Case I} ${(\tau \leq \mu _{1}):}$ In this case, the detection threshold is less than the minimum background noise at $W$, so we have $\mathcal{P}_{FA}^{(1)}=1$, $\mathcal{P}_{MD}^{(1)}=0$. Thus, the DEP is given by
%when the given detection threshold is less than the minimum background noise at $W$, the FA is bound to occur at $W$ (i.e., $\mathcal{P}_{FA}^{(1)}=1$, $\mathcal{P}_{MD}^{(1)}=0$). 
\begin{equation}
\begin{aligned}
 {\mathcal{P}_e}^{(1)}(\tau)=1.\label{case1}
\end{aligned}
\end{equation}
%\textbf{Case 2}$\bm {(\mu_1  \leq \tau \leq \mu_2):}$

%\emph{Case II} ($\mu_1  < \tau \leq P_w+\mu_2$ \text {and} $\mu_2<P_w+\mu_1$): In this case, as shown in Fig. 2(b), when a given detection threshold is less than $\mu_2$, FA occurs at $W$, when $H_0$ is true. While the detection threshold is larger than ${P_{\fff w}}+\mu_1$, $W$ will make an error of MD, when $H_1$ is true. There is no error in the yield $\tau \in (\mu_2, P_w+\mu_1)$.  Thus, the detection error probability is given by
%In this case, when a given detection threshold is between the bounded noise uncertainty. It means when the detected power is less than $ \mu_2$, $W$ will make an error of FA; when the detected power is more than ${P_{\fff w}}+\mu_1$, $W$ may make an error of MD, as shown in Fig.2 (b). So the detection error probability is given by
\emph{Case II} ($\mu_1  < \tau \leq \mu_2$): In this case, the detection threshold is between the two bounds of noise uncertainty, so FA and MD will occur. According to (\ref{pfa}) and (\ref{pmd}), the DEP can be determined as
%when a given detection threshold is between the bounded noise uncertainty, FA and MD will occur. 
\begin{flalign}
\begin{aligned}
{\mathcal{P}_e}^{(1)} (\tau) & =1- \left(\int_{0}^{\infty}\int_{\tau}^{\mu_2}f_{\sigma_{\fff W}^2}( x )f_{P_{\fff w}^{(1)}}(y )dxdy \right.\\
&\left. \quad\quad +\int_{0}^{\tau- \mu_1} \int_{\mu_1}^{\tau-y}f_{\sigma_n^2}(x)f_{P_{\fff w}^{(1)}}(y)dxdy\right)\\
&= \!1\!-\!\!\frac{1}{2\ln(\rho)}\! \left[\ln(\mu_2-\tau)+e^{-\frac{\tau}{P_{\fff S}}}\left( \mathbf{Ei}(\frac{\mu_1}{P_{\fff S}})\right.\right.\\
& \left.\left. \quad -\mathbf{Ei}(\frac{\tau}{P_{\fff S}})\right)\ln(\frac{\tau}{\mu_1}) \right]. \label{xi2}
 \end{aligned}
\end{flalign}

%{\mathcal{P}_e}^{(1)} (\tau) & =1- \left(\int_{0}^{\infty}\int_{\tau}^{\mu_2}f_{\sigma_{\fff W}^2}( x )f_{P_{\fff w}^{(1)}}(y )dxdy \right.\\
%&\left. \quad\quad +\int_{0}^{\tau- \mu_1} \int_{\mu_1}^{\tau-y}f_{\sigma_n^2}(x)f_{P_{\fff w}^{(1)}}(y)dxdy\right)\\
%&= \!1\!-\!\!\frac{1}{2\ln\rho}\! \ln \!\!\frac{\mu_2}{\tau}\!\!+\!\!\frac{1}{2\ln\rho {P_{\fff S}}^2}\!\!\ln \!\! \frac{1}{\mu_1}   \left[- \exp (- \frac{\tau-\mu_1}{ P_{\fff S}})\right.\\
%& \left.\quad \times \ln \!\mu_1+\ln\tau+ \mathbf{Ei}(- \frac{\tau-\mu_1}{ P_{\fff S}}) \right]. \label{xi2}

%\emph{Case III} {($\mu_1  < \tau \leq P_w+\mu_2$ \text {and} $\mu_2>P_w+\mu_1$):} In this case, as shown in Fig. 2(c), when a given detection threshold is less than $\mu_2$, FA occurs at $W$, when $H_0$ is true. While the detection threshold is larger than ${P_{\fff w}}+\mu_1$, $W$ will make an error of MD, when $H_1$ is true. The detection error probability always happens and is given by 
%In this case, when a given detection threshold is more than the maximum background noise at $W$, which means the detected power may small than the upper bound of the noise uncertainty as shown in Fig.2 (c), and thus the detection error probability is given by
\emph{Case III} {($\tau>\mu_2$):} In this case, when a given detection threshold is larger than the maximum background noise $\mu_2$ at $W$, so FA and MD occur.  According to (\ref{pfa}) and (\ref{pmd}), the DEP is given by
%more than the maximum background noise $\mu_2$ at $W$, FA and MD occurs.
\begin{flalign}
\begin{aligned}
 {\mathcal{P}_e}^{(1)} \!(\tau)\!\!& =1-\left(\int_{\tau-\!\mu_2}^{\tau-\mu_1} \int_{\mu_1}^{\tau-y} f_{\sigma_{\fff W}^2}(x)f_{P_{\fff w}^{(1)}}(y)dxdy \right.\\
 &\left. \quad\quad+  \int_{0}^{\tau-\mu_2} \int_{\mu_1}^{\mu_2} f_{\sigma_{\fff W}^2}(x)f_{P_{\fff w}^{(1)}}(y)dxdy\right)\\
&= \!\!1\!-\!\frac{1}{2\!\ln(\rho)}\left[ e^{-\frac{\tau}{P_{\fff S}}} \left( \mathbf{Ei}(\frac{\mu_1}{P_{\fff S}}+\mathbf{Ei}(\frac{\mu_2}{P_{\fff S}})-e^{\frac{\mu_1}{P_{\fff S}}}\ln{(\mu_1)}\right.\right.\\
&\left. \left. +e^{\frac{\mu_2}{P_{\fff S}}}\ln{(\mu_2)}\right)+1\right]. \label{xi3}
\end{aligned}
\end{flalign}
% &\left. \quad\quad+  \int_{0}^{\tau-\mu_2} \int_{\mu_1}^{\mu_2} f_{\sigma_{\fff W}^2}(x)f_{P_{\fff w}^{(1)}}(y)dxdy\right)\\
%&= \!\!1\!-\!\frac{1}{2\!\ln\rho {P_{\fff S}}^2}  \! \ln \!\! \frac{1}{\mu_1} \!\left [\!- \! \exp (\! -  \frac{\tau-\mu_1}{ P_{\fff S}}\!)\!\ln {\mu_1}\! +\!   \ln(\mu_2) \right.\\
%& \left.\quad \!\!\times \exp (\! - \frac{\tau-\mu_2}{P_{\fff S}} \! )\!+\! \mathbf{Ei}(-\frac{\tau-\mu_1}{ P_{\fff S}})\!-\! \mathbf{Ei}(-\frac{\tau-\mu_2}{P_{\fff S}})\right]\!\!.  \label{xi3}
%\emph{Case IV} ${(\tau>P_{w}+ \mu _{2}):}$ In this case, as shown in Fig. 2(d), when the given detection threshold is more than the maximum received power at $W$, the MD is bound to occur at $W$ (i.e., $\mathcal{P}_{MD}^{(1)}=1$, $\mathcal{P}_{FA}^{(1)}=0$). Thus, the detection error probability is given by
Finally, substituting (\ref{case1}), (\ref{xi2}), and (\ref{xi3}) into the ${\mathcal{P}_e}^{(1)}$ in (\ref{pe1-2}) completes the proof.
\end{proof}
\begin{comment}
\begin{figure}[htbp]
 \centering
    \subfigure[Case I ($\tau \leq \mu _{1}$) ]{
      \centering
        \includegraphics[width=2.9in]{figc.png}
    }
      \quad    %用 \quad 来换行
    \subfigure[Case II ($\mu_1  < \tau \leq P_w+\mu_2$ \text {and} $\mu_2<P_w+\mu_1$) ]{
        \includegraphics[width=2.9in]{figa.eps}
    }
      \quad    %用 \quad 来换行
    \subfigure[Case III ($\mu_1  < \tau \leq P_w+\mu_2$ \text {and} $\mu_2>P_w+\mu_1$)]{
	\includegraphics[width=2.9in]{figb.eps}
    }
    \quad    %用 \quad 来换行
    \subfigure[Case IV ($\tau>P_{w}+ \mu _{2}$) ]{
        \includegraphics[width=2.9in]{fige.png}
    }
    \flushright
	\includegraphics[width=0.73in]{figd.png}\par
    \caption{The detection error probability $\mathcal{P}_e$ of warden in each time slot.}
    \label{fig.2}
\end{figure}
\end{comment}
Based on a proof similar to that \textbf{Lemma 1}, we can prove that the DEP ${\mathcal{P}_e}^{(2)}$ in the second time slot is also determined as (\ref{pe1}) by replacing $P_{\fff S}$ with $P_{\fff R}$ only.

Regarding the optimal detection threshold and corresponding the minimum DEP at warden, we have the following lemma.
%We consider the covert performance for the worst-case of the concerned system, which means that warden deploys an optimal detection threshold $\tau^*$ to minimize his detection error probability. Thus, the minimum detection error probability is given by the following lemma.
%To obtain the minimum performance of $W$'s detection, we first optimum detection threshold $\tau$, and then the minimization detection error probability for $W$ under single/multiple antenna models is determined. So we present the following lemma.
\begin{lemma}
For both the single/multiple antenna scenarios of the concerned two-hop relay system with transmit power $P_{\fff S}$ for $S$, transmit power $P_{\fff R}$ for $R$, noise uncertainty $\rho$ and nominal noise power $\sigma_n^2$, the optimal detection threshold $\tau^*$, the optimal DEP ${\mathcal{P}_e}^{(1)*}$ in the first time slot and the optimal DEP ${\mathcal{P}_e}^{(2)*}$ in the second time slot are determined as
%For the concerned two-hop relay system with transmit power $P_{\fff S}$ for $S$, noise uncertainty $\rho$ and nominal noise power $\sigma_n^2$ of the system, the optimal detection threshold can be given by $\tau^*=\mu_2$, based on which the optimal detection error probability  ${\mathcal{P}_e}^{(1)*}$ in the first time slot of single/multiple antenna scenarios is given by
%We use ${\mathcal{P}_e}^{(1)*}$ to denote the minimum detection error probability of $W$, the optimal threshold $\tau^*$ for $W$ in the first time slot $\tau^* = \mu_2$, and the corresponding ${\mathcal{P}_e}^{(1)*}$ is given as 
 \begin{equation}
\begin{aligned}
\tau^*=\rho \sigma^2_n.\label{op-tau}
\end{aligned}
\end{equation}

\begin{equation}
\begin{aligned}
\!\!{\mathcal{P}_e}^{(1)*}\!= \!1\!-\!\!\frac{1}{2\ln(\rho)}\! \left[e^{-\frac{\mu_2}{P_{\fff S}}}\!\left( \mathbf{Ei}(\frac{\mu_1}{P_{\fff S}}) \!-\!\mathbf{Ei}(\frac{\mu_2}{P_{\fff S}})\right)\!\ln(\frac{\mu_2}{\mu_1}) \!\right]\!.
\label{m-pe1}
\end{aligned}
\end{equation}

\begin{flalign}
\begin{aligned}
\!\!{\mathcal{P}_e}^{(2)*}\!= \!1\!-\!\!\frac{1}{2\ln(\rho)}\! \left[e^{-\frac{\mu_2}{P_{\fff R}}}\!\left( \mathbf{Ei}(\frac{\mu_1}{P_{\fff R}}) \!-\!\mathbf{Ei}(\frac{\mu_2}{P_{\fff R}})\right)\!\ln(\frac{\mu_2}{\mu_1}) \!\right]\!.
\label{2-pe}
\end{aligned}
\end{flalign}
\end{lemma}
%{\mathcal{P}_e}^{(1)*}& =1-\frac{1}{2\ln\rho {P_{\fff S}}^2} \ln\frac{1}{\mu_1} \left [- \exp (-\frac{\mu_2-\mu_1}{ P_{\fff S}})\ln{\mu_1}\right.\\
%&\left.\quad + \ln(\mu_2)+ \mathbf{Ei}(-\frac{\mu_2-\mu_1}{ P_{\fff S}})\right].
\begin{proof}
We first focus on the proof for the first time slot. The minimum DEP ${\mathcal{P}_e}^{(1)*}$ of $W$ can be defined as
%This work considers the worst-case for covert communication, where $W$ aims to maximize his detection performance so that he can minimize the detection error probability to find the optimal threshold. Thus, the minimum detection error probability of $W$ can be defined as 
\begin{equation}
\begin{aligned}
 {\mathcal{P}_e}^{(1)*}=  \min \limits_{\tau^*}  {\mathcal{P}_e}^{(1)}(\tau).
\end{aligned}
\end{equation}
Note that the DEP in (\ref{pe1}) is a continuous bounded function. When $\mu_1<\tau\leq \mu_2 $, the first-order derivation of ${\mathcal{P}_e}^{(1)}$ in (\ref{xi2}) respect to $\tau$ is given by
 \begin{equation}
\begin{aligned}
{\mathcal{P}_e}^{(1)'}(\tau)=&\frac{1}{2\rho\tau(\tau-\mu_2)}\!\left[-e^{-\frac{\tau}{\rho}} \ln(\rho) \left((\mu_2-\tau) \left((\mathbf{Ei}(\frac{\mu_1}{P_{\fff S}})\right.\right.\right.\\
& \left.\left.\left. \!\!\!-\mathbf{Ei}(\!\frac{\tau}{P_{\fff S}}\!)\!)(\rho\!+\!\tau \ln (\!\frac{\tau}{\mu_1}\!)\!)\!-\!\rho e^{\frac{\tau}{P_{\fff S}}} \ln (\!\frac{\tau}{\mu_1}\!)\!\right)\!-\!\rho \tau e^{\frac{\tau}{\rho}}\!\right)\!\right].
\end{aligned}
\end{equation}
We define $f(x_1)=\mathbf{Ei}(x_1)$, then $f^{'}(x_1)=\frac{1}{x_1}e^{x_1}>0$ when $x_1>0$. Thus, $f(x_1)$ monotonically increases with $x_1$ when $x_1>0$. Since $\frac{\tau}{P_{\fff S}}>\frac{\mu_1}{P_{\fff S}}>0$, we know that $\mathbf{Ei}(\frac{\mu_1}{P_{\fff S}})-\mathbf{Ei}(\frac{\tau}{P_{\fff S}})<0$. Notice that $0<\mu_1<\tau<\mu_2$, then we have ${\mathcal{P}_e}^{(1)'}<0$.\par
When $\tau>\mu_2$, the first-order derivation of ${\mathcal{P}_e}^{(1)}$ in (\ref{xi3}) respect to $\tau$ is given by
\begin{equation}
\begin{aligned}
{\mathcal{P}_e}^{(1)'}(\tau)=&\frac{1}{2P_{\fff S}\ln(\rho)}e^{-\frac{\tau}{P_{\fff S}}} \left[\mathbf{Ei}(\frac{\mu_1}{P_{\fff S}})+\mathbf{Ei}(\frac{\mu_2}{P_{\fff S}})-e^{\frac{\mu_1}{P_{\fff S}}} \ln (\mu_1)\right.\\
&\left. +e^{\frac{\mu_2}{P_{\fff S}}}\ln (\mu_2)+e^{\frac{\mu_2}{P_{\fff S}}}\right].
\end{aligned}
\end{equation}
We define $f(x_2)=e^{\frac{x_2}{P_{\fff S}}}\ln (x_2)$, then $f^{'}(x_2)=e^{\frac{x_2}{P_{\fff S}}}(\frac{1}{P_{\fff S}}\ln (x_2)+\frac{1}{x_2})>0$ when $x_2>0$. Thus, $f(x_2)$ monotonically increases with $x_2$ when $x_2>0$. Since ${\mu_2}>{\mu_1}>0$, we know that $e^{\frac{\mu_2}{P_{\fff S}}}\ln (\mu_2)-e^{\frac{\mu_1}{P_{\fff S}}}\ln (\mu_1)>0$. Notice that $\tau>\mu_2$, then we have ${\mathcal{P}_e}^{(1)'}>0$.\par
The above results indicate that ${\mathcal{P}_e}^{(1)}$ monotonically decreases with $\tau$ when $\tau \in (\mu_1, \mu_2]$, and monotonically increases when $\tau \in ( \mu_2, \infty )$, so the optimal detection threshold $\tau^{*}$ is given by (\ref{op-tau}). 

By substituting the $\tau^*$ in (\ref{op-tau}) into (\ref{xi2}), we then know that the minimum DEP ${\mathcal{P}_e}^{(1)*}$  in the first time slot is determined as (\ref{m-pe1}). Based on a similar proof as above, we can see that in the second time slot, the optimal detection threshold $\tau^*$ and the DEP ${\mathcal{P}_e}^{(2)*}$ are given by (\ref{op-tau}) and (\ref{2-pe}), respectively.
\end{proof}
%Thus, ${\mathcal{P}_e}^{(1)}$ monotonically decreases with  Based on this observation, we can obtain the optimal detection threshold $\tau^{*}$, which is given by
%and (\ref{xi3}) with $\tau$ yields, we can obtain $ {\partial {\mathcal{P}_e}^{(1)}}/{\partial \tau} < 0$ and $ {\partial {\mathcal{P}_e}^{(1)}}/{\partial \tau} > 0$. 
%$ {\partial {\mathcal{P}_e}^{(1)}}/{\partial \tau} \leq 0$
%since the detection error probability at $W$ is a continuous bounded function, differentiating (22) and (23) with $\tau$ yields, we can obtain $ {\partial \xi}/{\partial \tau} \leq 0$. So the optimal detection threshold is the upper bound of noise uncertainty, which is given

Based on the results in \textbf{Lemma 2}, we can establish the following theorem on the minimum DEP of the two-hop relay system.
%Based on the analysis of the minimum detection error probability in the first time slot, we then derive it in the second time slot of two-hop relay system as ${\mathcal{P}_e}^{(2)*}$. According to \textbf{Lemma 1} and \textbf{Lemma 2}, we can derive ${\mathcal{P}_e}^{(2)*}$ and finally obtain the minimum detection error probability of two-hop relay system  $\xi^*$ as the following theorem.
%Since the considered communication scenario is a two-hop relay network, the minimum detection error probability ${\mathcal{P}_e}^{(2)*}$ of $W$ in the second time slot and the total minimum detection error probability $\xi^*$ of the two-hop relay should be considered. According to \textbf{Lemma 1} and \textbf{Lemma 2}, we can derive ${\mathcal{P}_e}^{(2)*}$ and then obtain $\xi^*$ as the following theorem.
\begin{theorem}
For both the single/multiple antenna scenarios of the concerned two-hop relay system with transmit power $P_{\fff S}$ for $S$, transmit power $P_{\fff R}$ for $R$, noise uncertainty $\rho$ and nominal noise power $\sigma_n^2$, the minimum DEP $\xi^*$ is given by
%For the concerned two-hop relay system with transmit power  of the system, the minimum detection error probability $\xi^*$ of the two-hop relay system under single/multiple antenna scenarios can be given by
\begin{equation}
\begin{aligned}
\xi^* &\!\!=\!1-\!\left[ \frac{1}{4(\ln\rho)^2}(\ln\frac{\mu_2}{\mu_1})^2\!\left[e^{-\frac{\mu_2}{P_{\fff S}}}\!\!\left( \mathbf{Ei}(\frac{\mu_1}{P_{\fff S}}) \!-\!\mathbf{Ei}(\!\frac{\mu_2}{P_{\fff S}})\!\right)\!\right]\!\!\left[e^{-\frac{\mu_2}{P_{\fff R}}}\right.\right.\\
& \quad \left. \left. \times\left( \mathbf{Ei}(\frac{\mu_1}{P_{\fff R}}) \!-\!\mathbf{Ei}(\frac{\mu_2}{P_{\fff R}})\right)\right]\right].
\label{two-hop dep}
\end{aligned}
\end{equation}
\end{theorem}\par
%\xi^* &\!=\!\frac{1}{4\ln\!\rho {P_{\fff R}}^2\!{P_{\fff S}}^2}\! \ln\!\frac{1}{\mu_1}\!\!\left[\!\!- \exp (\!-\frac{\mu_2-\mu_1}{ P_{\fff S}})\!\ln{\mu_1}\!+\! \ln\mu_2 \right.\\
%&\left.\quad+ \mathbf{Ei}(-\frac{\mu_2-\mu_1}{ P_{\fff S}})\right]  \left[- \exp (-\frac{\mu_2-\mu_1}{ P_{\fff R}})\ln{\mu_1}\right.\\
%&\left.\quad+ \ln(\mu_2)+ \mathbf{Ei}(-\frac{\mu_2-\mu_1}{ P_{\fff R}})\right].\label{two-hop dep}
\begin{proof}
%Similar as \textbf{Lemma 1} and \textbf{Lemma 2}, the minimum detection error probability $ {\mathcal{P}_e}^{(2)*}$ in the second time slot can be determined as
% {\mathcal{P}_e}^{(2)*}& =\frac{1}{2\ln\rho {P_{\fff R}}^2} \ln\frac{1}{\mu_1} \left[- \exp (-\frac{\mu_2-\mu_1}{ P_{\fff R}})\ln{\mu_1}\right.\\
%&\left. \quad + \ln(\mu_2)+ \mathbf{Ei}(-\frac{\mu_2-\mu_1}{ P_{\fff R}})\right].
Due to the assumption that the detection in two time slots is independent, the covert performance in each transmission needs to be guaranteed. According to (\ref{2-dep}), the minimum DEP $\xi^*$ of the two-hop relay system can be determined as
%which means that the transmissions are not detected in both two time slots.
%As we assumed in the single/multiple antenna scenarios of the two-hop relay, the first time slot and the second time slot are independent. In order to achieve covert communication in the two-hop wireless communication system, we have to ensure the covertness of the first time slot as well as the second time slot. In other words, under the condition that the first time slot is not detected by $W$, and the second time slot is not detected, it can realize the covert communication from the transmitter to the receiver through a relay in two-hop wireless communications as shown in Tab.1. Thus, the minimum detection error probability of $W$ in this two-hop relay system can be determined as
\begin{equation}
\begin{aligned}
{ \xi}^*= 1-(1- {\mathcal{P}_e}^{(1)*})(1- {\mathcal{P}_e}^{(2)*}).\label{2-mdep}
\end{aligned}
\end{equation}
By substituting the expressions of ${\mathcal{P}_e}^{(1)*}$ and ${\mathcal{P}_e}^{(2)*}$ into (\ref{2-mdep}), we can see that $\xi^*$ is determined by (\ref{two-hop dep}). 
\end{proof}

\subsection{Covert Throughput in Single-antenna Scenario}
The covert throughput is defined as the average rate of successfully transmitted covert messages from $S$ to $D$ \cite{he2018covert}, \cite{zheng2019multi}. If we use $T$ to denote the target transmission rate of $S$ and use $p_{out}$ to denote the transmission outage probability of the covert communication from $S$ to $D$, then the covert throughput $\eta$ of the concerned relay system is evaluated as
%We are ready to derive the covert throughput, defined as the average successfully transmitted amount of covert messages from $S$ to $D$, which can be expressed as the product of the connectivity probability (transmission without outage probability) and the target rate\cite{he2018covert},\cite{zheng2019multi}. Thus, the covert throughput can be given by 
 \begin{equation}
\begin{aligned}
 \eta = T(1-p_{out}),\label{d_eta}
\end{aligned}
\end{equation}
%where $T$ denotes the target rate of messages adopted at $S$, and $p_{out}$ denotes the transmission outage probability of the covert communication from $S$ to $D$.
%defined as the throughput of covert messages sent between legitimate nodes, which is related to the target rate and the transmission outage probability. While the transmission outage probability is related to the relationship between channel capacity and the target rate. \par
%We first analyze the transmission outage probability $p_{out}^{sin}$ of two-hop relay system in single-antenna scenario. \par
%based on (\ref{s-csr}) and (\ref{s-crd})

Regarding the covert throughput of the concerned relay system under single-antenna scenario, we have the following theorem.
%Based on (\ref{d_eta}), we derive the covert throughput in single-antenna and multi-antenna scenarios, respectively. The covert throughput $\eta_{sin}$ in the single-antenna scenario can be given by the following theorem.
%Based on the transmission outage probability and the pre-determined target rate of the system, we can obtain the covert throughput by the following theorem.
\begin{theorem}
For the single-antenna scenario of the concerned two-hop relay system with transmit power $P_{\fff S}$ for $S$, transmit power $P_{\fff R}$ for $R$, noise uncertainty $\rho$, nominal noise power $\sigma_n^2$ and target rate $T_s$, the corresponding covert throughput $\eta_{s}$ is given by
\begin{equation}
\begin{aligned}
\eta_{s}&=T_{s}\!\left[\frac{1}{4(\ln\!\rho)^2}\mathbf{Ei}\left(\frac{\kappa_1\! \rho \sigma_n^{2}}{P_{\fff S}},\!\frac{\kappa_1\!\sigma_n^{2}}{P_{\fff S} \rho}\right)\! \times\!\mathbf{Ei}\left(\frac{\kappa_1\!\rho \sigma_n^{2}}{P_{\fff R}},\!\frac{\kappa_1\!\sigma_n^{2}}{P_{\fff R} \rho}\right)\right]\!,
\label{eta_s}
\end{aligned}
\end{equation}
%where  $P_{sin}$ is the optimal transmit power at the transmitter (i.e., $S$ or $R$) and it is given by%$\kappa=2^{2R_{sin}}-1$, $P_{sin} = \min(P_{\fff S}, P_{\fff R})$, $\mathbf{Ei}(x,y) =\mathbf{Ei}(-x)-\mathbf{Ei}(-y)$. \par
%\begin{equation}
%\begin{aligned}
%P_{sin} = \min(P_{\fff S}, P_{\fff R}),\label{p_sin}
%\end{aligned}
%\end{equation}
where $\kappa_1=2^{2T_{s}}-1$ and $\mathbf{Ei}(x,y) =\mathbf{Ei}(-x)-\mathbf{Ei}(-y)$. \par
\end{theorem}

\begin{proof} From (\ref{d_eta}) we can see that to evaluate the covert throughput $\eta_{s}$ in the single-antenna scenario, we need to first determine the outage probability $p_{out}^{(s)}$ in this scenario. Due to the DF relaying concerned in this work, $p_{out}^{(s)}$ can be determined as
%first determine the outage probability in the single-antenna scenario to evaluate the covert throughput $\eta_{s}$. Due to the DF relaying concerned in this work, the outage probability $p_{out}^{(s)}$ of the relay system can be determined as 
\begin{equation}
\begin{aligned}
p_{out}^{(s)} &=1-(1-p_{out}^{(s,1)})(1-p_{out}^{(s,2)}),
\label{pout_sin}
\end{aligned}
\end{equation}
where $p_{out}^{(s,1)}$ and $p_{out}^{(s,2)}$ denote the outage probability in the first hop (time slot) and second hop (time slot), respectively. The outage probability $p_{out}^{(s,1)}$ in the first hop is defined as
\begin{equation}
\begin{aligned}
p_{out}^{(s,1)} = \mathbb{P}(C_{\fff{SR}}^{(s)}<T_{s}),\label{pout_sin_1}
\end{aligned}
\end{equation}
here $C_{\fff SR}^{(s)}$ is the capacity of $S \rightarrow R$ channel in this scenario, which can be determined as
\begin{equation}
\begin{aligned}
C_{\fff SR}^{(s)}= \frac{1}{2}  \log _2(1+ \frac{P_{\fff S}\mid h_{\fff SR}\mid ^{2}}{\sigma_{\fff SR}^2}). \label{s-csr}
\end{aligned}
\end{equation}
Notice that the noise power $\sigma_{\fff SR}^2$ at $R$ follows a logarithmic uniform distribution \cite{he2017covert}, so its PDF $f_{\sigma_{\fff SR}^2}(x)$ can be determined as
\begin{equation}
\begin{aligned}
f_{ \sigma_{\fff SR}^{2}}(x)=  \begin{cases} \frac{1}{2x\ln (\rho) } , & \text{if}\  \frac{\sigma_n^2}{\rho}\leq x \leq \rho\sigma_n^2 \\0, & \text{otherwise}.\end{cases}\label{pdf-np}
\end{aligned}
\end{equation}
The channel gain ${\mid h_{\fff SR}\mid^2}$ from $S$ to $R$ can be modeled as exponential distribution \cite{gao2021covert}, so its PDF $f_{|h_{\fff SR}|^2}(x)$ is determined as
\begin{equation}
\begin{aligned}
f_{\mid h_{\fff SR}\mid^2}(x)=\exp(-x).
\label{pdf_hsr}
\end{aligned}
\end{equation}
Based on (\ref{pout_sin_1})-(\ref{pdf_hsr}), we know that the outage probability $p_{out}^{(s,1)}$ in the first hop can be derived as
\begin{equation}
\begin{aligned}
p_{out}^{(s,1)} &= \mathbb{P}\left(\frac{P_{\fff S}\mid h_{\fff SR}\mid ^{2}}{\sigma_{\fff SR}^2}<2^{2T_{s}}-1 \right)\\
&=\int_{\frac{\sigma_n^2}{\rho}}^{\rho\sigma_n^2} \! \int_{0}^{\frac{(2^{2T_{s}}-1)y}{P_{\fff S}}} \!\! f_{\mid h_{\fff SR}\mid^2}(x)f_{\sigma_{\fff SR}^2}(y)dxdy\\
&=\!1\!-\!\frac{1}{2\ln  \rho}\!\!\left[\!\mathbf{Ei}(-\frac{(2^{2T_{s}}\!-\!1) \rho \sigma_n^{2}}{P_{\fff S}})\!-\!\mathbf{Ei}(-\frac{(2^{2T_{s}}\!-\!1)\sigma_n^{2}}{P_{\fff S} \rho})\!\right]\!.
\label{p_out_s1}
\end{aligned}
\end{equation}
Following a similar analysis as above, we can see that the outage probability $p_{out}^{(s,2)}$ in the second hop is given by
\begin{equation}
\begin{aligned}
p_{out}^{(s,2)} &= \mathbb{P}\left(\frac{P_{\fff R}\mid h_{\fff RD}\mid ^{2}}{\sigma_{\fff RD}^2}<2^{2T_{s}}-1 \right)\\
&=\int_{\frac{\sigma_n^2}{\rho}}^{\rho\sigma_n^2} \! \int_{0}^{\frac{(2^{2T_{s}}-1)y}{P_{\fff R}}} \!\! f_{\mid h_{\fff RD}\mid^2}(x)f_{\sigma_{\fff RD}^2}(y)dxdy\\
&=\!1\!-\!\frac{1}{2\ln  \rho}\!\!\left[\!\mathbf{Ei}(-\frac{(2^{2T_{s}}\!-\!1) \rho \sigma_n^{2}}{P_{\fff R}})\!-\!\mathbf{Ei}(-\frac{(2^{2T_{s}}\!-\!1)\sigma_n^{2}}{P_{\fff R} \rho})\!\right]\!.
\label{p_out_s2}
\end{aligned}
\end{equation}
By substituting (\ref{p_out_s1}) and (\ref{p_out_s2}) into (\ref{pout_sin}), the outage probability $p_{out}^{(s)}$ of the system in single-antenna scenario can be determined as
 \begin{equation}
\begin{aligned}
\!\!p_{out}^{(s)}\!\!&=\!1\!-\!\frac{1}{4\!(\ln  \rho)^2}\!\!\left[\!\mathbf{Ei}(\!-\frac{(2^{2\!T_{s}}\!-\!1\!) \rho \sigma_n^{2}}{P_{\fff S}})\!-\!\mathbf{Ei}(\!-\frac{(2^{2\!T_{s}}\!-\!1)\sigma_n^{2}}{P_{\fff S} \rho}\!)\!\right]\!\!\\
& \quad \times \left[\!\mathbf{Ei}(\!-\frac{(2^{2\!T_{s}}\!-\!1\!) \rho \sigma_n^{2}}{P_{\fff R}})\!-\!\mathbf{Ei}(\!-\frac{(2^{2\!T_{s}}\!-\!1)\sigma_n^{2}}{P_{\fff R} \rho}\!)\!\right].\!\!
\label{pout_40}
%&=1-\frac{1}{2\ln  \rho}\mathbf{Ei}(\frac{(2^{2\!R_{sin}}\!-\!1\!) \rho \sigma_n^{2}}{P_{sin}},\frac{(2^{2\!R_{sin}}\!-\!1\!)\sigma_n^{2}}{P_{sin} \rho}).
\end{aligned}
\end{equation}
Finally, by substituting $T_s$ and (\ref{pout_40}) into (\ref{d_eta}), we can see that  $\eta_{s}$ is determined by (\ref{eta_s}).
\end{proof}

\subsection{Covert Throughput in Multi-antennas Scenario}
Regarding the throughput of the concerned relay system under the multi-antenna scenario, we have the following theorem.
%In this subsection, we derive the covert throughput of two-hop relay system in multi-antenna scenario. Because the definition of covert throughput in multi-antenna scenario is the same as that in single-antenna scenario, thus, the covert throughput of the concerned relay system under multi-antenna scenario, we have the following theorem.\par

\begin{theorem}
For the multi-antenna scenario of the concerned two-hop relay system with transmit power $P_{\fff S}$ for $S$, transmit power $P_{\fff R}$ for $R$, noise uncertainty $\rho$, nominal noise power $\sigma_n^2$ and target rate $T_m$, the corresponding covert throughput $\eta_{m}$ is given by (\ref{mul-eta}),
%吞吐量长公式
\begin{figure*}[t!]
 \normalsize
 %\hrulefill
	\vspace*{4pt}
\small
 \begin{equation}
 \begin{aligned}
\!\!\eta_{m} \!&=\!T_{m}\!\left\{1-\!\!\frac{N_{\fff S}}{2\ln \rho(N_{\fff R_{r}}-\!1)!}\!\! \!\! \sum\limits_{s=0}^{N_{\fff S}-1}(-1)^s\!\!\left(\begin{array}{c}{\!\!\!\!N_{\fff S}-1\!\!\!\!}\\ s\end{array}\right)\!\! \! \sum\limits_{k=0}^{s(N_{\fff S}-1)}\!\!\! \left[ \sum\limits_{j=0}^{\omega_1}\frac{j!\left(\begin{array}{c}{\!\!\!\! \omega_1\!\!\!\!}\\ j\end{array}\right)}{(s+1)^{j+1}}\!\! \left( \begin{array}{c}{\!\!\!\!\kappa_2\!\!\!\!}\\ P_{\fff S}\end{array}\right)^{\omega_1-j}\!\! \!\! \kappa_3^{-(N_{\fff R_{r}}-1)} \left[\Gamma(N_{\fff R_{r}}-1,\kappa_3\mu_2) \right.\right.\right.\\
& \quad \left.\left.\left. -\Gamma(N_{\fff R_{r}}-1,\kappa_3\mu_1)\right]\! +\! \frac{\omega_1\ln \rho^2}{(s+1)^{N_{\fff R_{r}}+k}}\! \right]\!\right\}\!\!\left\{\!\!1-\!\!\frac{N_{\fff R_{t}}}{2\ln \rho(N_{\fff D}\!-\!1)!}\!\! \!\! \sum\limits_{r=0}^{N_{\fff R_{t}}-1}(-1)^r\!\!\left(\begin{array}{c}{\!\!\!\!N_{\fff R_{t}}-1\!\!\!\!}\\ r\end{array}\right)\!\! \! \sum\limits_{q=0}^{r(N_{\fff R_{t}}-1)}\!\!\! \left[\! \sum\limits_{n=0}^{\omega_2}\frac{n!\left(\begin{array}{c}{\!\!\!\! \omega_2\!\!\!\!}\\ n\end{array}\right)}{(r+1)^{n+1}}\!\! \left(\!\! \begin{array}{c}{\!\!\!\!\kappa_2\!\!\!\!}\\ P_{\fff R}\end{array}\!\!\right)^{\omega_2-n}\!\! \!\!  \right.\right.\\
& \quad \left.\left. \times \kappa_4^{-(N_{\fff D}-1)} \left[\Gamma(N_{\fff D}-1,\kappa_4\mu_2)-\Gamma(N_{\fff D}-1,\kappa_4\mu_1)\right]+ \frac{\omega_2\ln \rho^2}{(r+1)^{N_{\fff D}+q}} \right]\right\},
\label{mul-eta}
\end{aligned}	 
\end{equation}
	\hrulefill
	\vspace*{4pt}
\end{figure*}
where $\kappa_2=2^{2T_{m}}-1$, $\kappa_3=((j-1)(2^{2T_{m}}-1)/P_{\fff S})^{-(N_{\fff R_{r}}-1)}$, $\kappa_4=((n-1)(2^{2T_{m}}-1)/P_{\fff R})^{-(N_{\fff D}-1)}$, $\omega_1=N_{R_{r}}+k-1$, $\omega_2=N_{D}+q-1$ and $\Gamma(\cdot)$ denotes Gamma function.
\end{theorem}
%原吞吐量公式\!T_{m}\!\!\left\{1-\!\! \frac{1}{2\ln \rho}\frac{N_{t}}{(N_{r}-\!1)!}\!\! \!\! \sum\limits_{t=0}^{N_{t}-1}(-1)^t\!\!\left(\begin{array}{c}{\!\!\!\!N_{t}-1\!\!\!\!}\\ t\end{array}\right)\!\! \! \sum\limits_{r=0}^{t(N_{t}-1)}\!\!\! \left[ \sum\limits_{k=0}^{r+N_{r}-1}\frac{k!\left(\begin{array}{c}{\!\!\!\! \omega_1\!\!\!\!}\\ k\end{array}\right)}{(t+1)^{k+1}}\!\! \left( \begin{array}{c}{\!\!\!\!\kappa_2(R_{mul})\!\!\!\!}\\ P_{mul}\end{array}\right)^{\omega_1-k}\!\! \!\! \kappa_3(P_{mul})^{-(\omega_1-j)}\right.\right.\\
%&\quad\left.\left. \times \left[\Gamma(\omega_1 - r,\kappa_3(P_{mul})\mu_2) -\Gamma(\omega_1 - r,\kappa_3(P_{mul})\mu_1)\right] + \frac{\omega_1!\ln \rho^2}{(t+1)^{N_{r}+r}} \right]\right\},
\begin{proof} From (\ref{d_eta}) we can see that to evaluate the covert throughput $\eta_{m}$ in the multi-antenna scenario, we need to first determine the outage probability $p_{out}^{(m)}$ in this scenario. Due to the DF relaying, $p_{out}^{(m)}$ can be determined as 
%we need to first determine the outage probability in the multi-antenna scenario to evaluate the covert throughput $\eta_{m}$. Due to the DF relay network, the outage probability $p_{out}^{(m)}$ of the relay system can be determined as 
\begin{equation}
\begin{aligned}
p_{out}^{(m)} &=1-(1-p_{out}^{(m,1)})(1-p_{out}^{(m,2)}),
\label{pout_mul}
\end{aligned}
\end{equation}
where $p_{out}^{(m,1)}$ and $p_{out}^{(m,2)}$ denote the outage probability in the first hop and second hop, respectively. The outage probability $p_{out}^{(m,1)}$ in the first hop is defined as
\begin{equation}
\begin{aligned}
p_{out}^{(m,1)} = \mathbb{P}(C_{\fff{SR}}^{(m)}<T_{m}),\label{pout_m_1}
\end{aligned}
\end{equation}
here $C_{\fff SR}^{(m)}$ is the capacity of $S \rightarrow R$ channel in multi-antenna scenario. According to \cite{molisch2005capacity}, $C_{\fff SR}^{(m)}$ can be determined as 
\begin{equation}
\begin{aligned}
\!\!\!\!C_{\fff S\!R}^{(m)}\!&=\!\frac{1}{2}\!\log _2\!\mid \mathbf{I}_{N_{\fff R_{r}}}\!+\!\frac{P_{\fff S}}{{\sigma_{\fff SR}^2}}\mathbf{h}_{\fff {i^{*}R}}^{H}\mathbf{w}_1^{H}\mathbf{w}_1\mathbf{h}_{\fff {i^{*}R}}\!\mid \!\!\!\! \\
&=\!\frac{1}{2}\!\log _2\!\mid \mathbf{I}_{N_{\fff R_{r}}}\!+\!\frac{P_{\fff S}}{{\sigma_{\fff SR}^2}}||\mathbf{w}_1\mathbf{h}_{\fff {i^{*}R}}||^2\mid.
\label{csr-mul}
%C_{\fff SR}= \frac{1}{2}  \log _2(1+ \frac{\mathbf{SNR}}{N_{\fff S}} {\parallel  \mathbf{h_{\fff SR}}  \parallel  } _F^2 ).
\end{aligned}
\end{equation}
For simplicity, let $Z=||\mathbf{w}_1\mathbf{h}_{\fff {i^{*}R}}||$. Under TAS-MRC, the channel gain $|Z|^2$ from $S$ to $R$ can be rewritten as 
%Let $|Z|^2$ denote the channel coefficient $||\mathbf{w}_1\mathbf{h}_{\fff {i^{*}R}}||^2$ from $S$ to $R$ for TAS-MRC, which can be written as 
\begin{equation}
\begin{aligned}
|Z|^2 = \mathop{\max}_{1\leq i\leq N_{\fff S}}\left\{{h}_{\fff iR}= \sum_{j=1}^{N_{\fff R_{r}}}|h_{ij}|^2\right\}. \label{Z}
\end{aligned}
\end{equation}
Notice that the total channel gain ${h}_{i{\fff R}}$ from the $i$-th antenna to $R$ can be modeled as a central chi-squared random variable with $2N_{\fff R_{r}}$ degrees of freedom, so its PDF $f_{{h}_{{i\fff R}}}(x)$ and cumulative distribution function (CDF) $F_{{h}_{i{\fff R}}}(x)$ are determined in \cite{chen2005analysis, haykin1988digital} as 
%Notice that the total channel coefficient ${h}_{i{\fff R}}$ from the $i$-th antenna to R can be modeled as a central chi-squared random variable with $2N_{\fff R}$ degrees of freedom, so its PDF $f_{{h}_{{i\fff R}}}(x)$ is given by  \cite{haykin1988digital}
\begin{equation}
\begin{aligned}
f_{ {h}_{iR}}(x)=\frac{1}{(N_{\fff R_{r}}-1)!}x^{N_{\fff R_{r}}-1}e^{-x},
\label{pdf_hiR}
\end{aligned}
\end{equation}
 %and the cumulative distribution function (CDF) of $F_{{h}_{i{\fff R}}}(x)$ is determined as 
 \begin{equation}
\begin{aligned}
F_{ {h}_{iR}}(x)=1- e^{-x} \sum\limits_{j=0}^{N_{\fff R_{r}}-1} \frac{x^j}{j!}.
\label{cdf_hiR}
\end{aligned}
\end{equation}
Based on (\ref{pdf_hiR}) and (\ref{cdf_hiR}), the PDF $f_{\mid  {Z}\mid^2}(x)$ of channel gain $|Z|^2$ from the selected $i^*$-th antenna to $R$ can be determined by
%then the channel coefficient $|Z|^2$ from the $i^*$-th antenna to R can be determined by
\begin{equation}
\begin{aligned}
f_{\mid  {Z}\mid^2}(x) \!&=\! N_{\fff S}\left[F_{ {h}_{iR}}(x)\right]^{N_{\fff S}-1}f_{ {h}_{iR}}(x)\\
&=\!\frac{N_{\fff S}}{(N_{\fff R_{r}}-1)!}\left(1-e^{-x}\!\!\!\sum\limits_{j=0}^{N_{\fff R_{r}}-1}\!\! \frac{x^j}{j!} \right)^{N_{\fff S}-1}x^{N_{\fff R_{r}}-1}e^{-x}\!\!.\label{pdf-z}
\end{aligned}
\end{equation}
%Based on (\ref{pout_m_1}), (\ref{csr-mul}), (\ref{pdf-np}) and (\ref{pdf-z}), we know that the outage probability $p_{out}^{(m,1)}$ in the first hop can be derived as (\ref{pout_mul_1}), in which the Step (a) is obtained based on \cite[Eq. (3.351)]{gradshteyn2014table} and \cite[Eq. (8.310.1)]{gradshteyn2014table}. Following a similar analysis as above, we can see that the outage probability $p_{out}^{(m,2)}$ in the second hop is given by (\ref{pout_mul_2}). By substituting (\ref{pout_mul_1}) and (\ref{pout_mul_2}) into (\ref{pout_mul}), the outage probability $p_{out}^{(m)}$ of the system in multi-antenna scenario can be determined as (\ref{pout_m}).
Based on (\ref{pout_m_1}), (\ref{csr-mul}), (\ref{pdf-np}) and (\ref{pdf-z}), we know that the outage probability $p_{out}^{(m,1)}$ in the first hop can be derived as (\ref{pout_mul_1}), in which the Step (a) is obtained based on \cite[Eq. (3.351)]{gradshteyn2014table} and \cite[Eq. (8.310.1)]{gradshteyn2014table}. Following a similar analysis as above, we can derive the outage probability $p_{out}^{(m,2)}$ in the second hop. By substituting $p_{out}^{(m,1)}$ and $p_{out}^{(m,2)}$ into (\ref{pout_mul}), the outage probability $p_{out}^{(m)}$ of the system in multi-antenna scenario can be determined.

%Finally, by substituting $T_m$ and (\ref{pout_m}) into (\ref{d_eta}), we can see that $\eta_m$ is determined by (\ref{mul-eta}).
Finally, by substituting $T_m$ and $p_{out}^{(m)}$ into (\ref{d_eta}), we can see that $\eta_m$ is determined by (\ref{mul-eta}).
\end{proof}

\begin{figure*}[t!]
 \normalsize
% \hrulefill
	\vspace*{4pt}
	\small
 \begin{equation}
 \begin{aligned}
p_{out}^{(m,1)}\!&=\mathbb{P}\{C_{\fff SR}^{(m)} < 2^{2T_{m}}-1\}\\
	&=  \int_{\frac{\sigma_n^2}{\rho}}^{ \rho \sigma_n^2 }\int_{0}^{ \frac{(2^{2T_{m}}-1) x}{P_{\fff S}}}\!\! \frac{1}{2x\ln \rho}\frac{N_{\fff S}}{(N_{\fff R_{r}}-1)!} 
	\left(1-e^{-y}\sum\limits_{j=0}^{\!\!N_{\fff R_{r}}-1\!\!}\frac{y^{j}}{j!}\right)^{N_{\fff S}-1}y^{N_{\fff R_{r}}-1}e^{-y}dxdy\\
	%&=\!\! \frac{(j+N_{\fff R}-1)!P_{\fff S}^{j+N_{\fff R}}}{ [(i+1) \sigma ^2]^{j+N_{\fff R}}}\! \!\!\!\sum\limits_{i=0}^{N_{\fff S}-1}(-1)^i \!\!\left(\begin{array}{c}{\!\!\!\!N_{\fff S}-1\!\!\!\!}\\ i\end{array}\right) \!\!\!\!\! \sum\limits_{j=0}^{i(N_{\fff S}-1)} \!\!\!\!bt_{ij}( \frac{ \sigma _{\fff SR}^2}{P_{\fff S}} )^{j+N_{\fff R}} \!\![-\!\exp(- \frac{(j\!+\!\!N_{\fff R}\!\!-\!\!1)(i+1) \sigma ^2}{P_{\fff S}} ) \!\!\!\!\sum\limits_{k=0}^{j+N_{\fff R}-1}\!\!\!\!\!\! \frac{(j+N_{\fff R}-1)!}{n!}  \frac{(2^{2R_{\fff D}}-1)^kP_{\fff S}^{j+N_{\fff R}-k}}{((i+1) \sigma ^2)^{j+N_{\fff R}-k}}  ]\\
	%&=\!\! \frac{1}{2\ln \rho}\frac{N_{\fff S}}{(N_{\fff R}-1)!}\!\! \!\! \sum\limits_{i=0}^{N_{\fff S}-1}(-1)^i\!\! \!\! \left(\begin{array}{c}{\!\!\!\!N_{\fff S}-1\!\!\!\!}\\ i\end{array}\right)\!\! \!\! \sum\limits_{j=0}^{i(N_{\fff S}-1)}bt(N_{\fff R},i)\{ \sum\limits_{k=0}^{j+N_{\fff R}-1}\frac{k!\left(\begin{array}{c}{\!\!\!\!m\!\!\!\!}\\ k\end{array}\right)}{(i+1)^{k+1}}\left(\begin{array}{c}{\!\!\!\!\kappa_1\!\!\!\!}\\ P\end{array}\right)^{\kappa_2-k}\kappa_3^{-(\kappa_2-j)}[\Gamma(\kappa_2-j,\kappa_3\rho\sigma_n^2)-\Gamma(\kappa_2-j,\frac{\sigma_n^2\kappa_3}{\rho})]+\frac{m!\ln \rho^2}{(i+1)^{N_{\fff R}+j}}\}
	&\overset{(a)}{=}\!\! \frac{1}{2\ln \rho}\frac{N_{\fff S}}{(N_{\fff R_{r}}-\!1)!}\!\! \!\! \sum\limits_{s=0}^{N_{\fff S}-1}(-1)^s\!\!\left(\begin{array}{c}{\!\!\!\!N_{\fff S}-1\!\!\!\!}\\ s\end{array}\right)\!\! \! \sum\limits_{k=0}^{s(N_{\fff S}-1)}\!\!\! \left\{ \sum\limits_{j=0}^{k+N_{\fff R_{r}}-1}\frac{j!\left(\begin{array}{c}{\!\!\!\! k+N_{\fff R_{r}}-1\!\!\!\!}\\ j\end{array}\right)}{(s+1)^{j+1}}\!\! \left( \begin{array}{c}{\!\!\!\!2^{2T_{m}}-1\!\!\!\!}\\ P_{\fff S}\end{array}\right)^{N_{\fff R_{r}}+k-j-1}\!\! {\frac{(2^{2T_{m}}-1)}{P_{\fff S}}}^{-(N_{\fff R_{r}}-1)}\right.\\
&\quad \times \left. (j-1)\left[\Gamma(N_{\fff R_{r}}-1,\frac{(j-1)(2^{2T_{m}}-1)}{P_{\fff S}}\mu_2)\! -\! \Gamma(N_{\fff R_{r}}-1,\frac{(j-1)(2^{2T_{m}}-1)}{P_{\fff S}}\mu_1)\right]\! +\! \frac{(k+N_{\fff R_{r}}-1)\ln \rho^2}{(s+1)^{N_{\fff R_{r}}+k}}\!\! \right\}.\label{pout_mul_1}\\ 
\end{aligned}		 
\end{equation} 
	\hrulefill
	%\vspace*{4pt}
\end{figure*}

\section{Covert Throughput maximization}
In this section, we formulate the optimization problems for covert throughput maximization under both single and multi-antenna scenarios. We determine the analytical conditions for the optimal parameter settings in the single antenna scenario to achieve the maximal covert throughput there, and also develop a direct search algorithm for covert throughput maximization in the more complex multi-antenna scenario.
%In this section, to maximize the covert throughput of two-hop relay system, we first formulate the optimization problems for both single and multi-antenna scenarios. Then we analyze the optimization problem in single antenna scenario and develop an alternative search algorithm to solve the optimal problem in multi-antenna scenario.
%In this section, we first formulate the optimization problems for single/multiple antenna scenarios to maximize the covert throughput of two-hop relay system. 
%to maximize the covert throughput of two-hop relay system in single/multiple antenna scenarios, 
%to obtain the maximum covert throughput in single and multiple antenna models of two-hop relay system, we propose a joint optimal design of the transmit power and transmission target rate algorithm, so that we can gain the maximal covert throughput of the system subject to the constraint of covertness, reliability and maximum transmit power.  
%Equations (34) and (42) can be used to define the optimal design of the single antenna model and multi-antenna model, respectively. Since there isn't a closed form for the expression, it is difficult to design the optimal solution with the convex optimization algorithm. The calculation can be reduced based on the monotonicity of the expression in order to reduce computing complexity and search for the maximal solution for the covert throughput under the three constraints.
\subsection{Covert Throughput Maximization in Single-antenna Scenario}
From \textbf{Theorem 2} we can see that the covert throughput is related to $P_{\fff S}$, $P_{\fff R}$ and ${T_s}$. Here, we determine the optimal settings of these parameters to achieve the maximum covert throughput $\eta_{s}^*$ in the single-antenna scenario. This optimization problem can be formulated as follows.
%To achieve the optimal covert performance, the covert throughput in \textbf{Theorem 2} is subject to covertness, reliability and transmit power constraints, which is defined as the maximum covert throughput of two-hop relay system in single-antenna scenario $\eta_{s}^*$ and is given by
%Based on the Theorem 2, the covert throughput $ \eta_{sin}$ with the covertness constraint, reliability constraint and maxim transmit power constraint in single antenna model of two-hop relay system can be given by
\begin{subequations}
 \begin{align}
%&\!\!\!\!\!\!\! \max_{P_{sin},R_{sin}}  \eta_{sin} =(1-p_{out}^{sin}(P_{sin}, R_{sin})) R_{sin},  \\
%\eta_{s}^* = &\max_{P_{\fff S}, P_{\fff R}, T_{s}}(1-p_{out}^{(s)}) T_{s},\label{sin_eta_a}  \\
&\max_{P_{\fff S}, P_{\fff R}, T_{s}}\eta_{s}  \label{sin_eta_a}\\
& \quad s.t. \quad \xi ^*\geq 1- \epsilon ,   \label{sin_eta_b}  \\
 &\!\!\!\qquad\qquad  p_{out}^{(s)} \leq  \delta ,  \label{sin_eta_c}  \\
&\qquad\quad\  P_{\fff S} \leq P_{ max },\label{sin_eta_d}  \\
&\qquad\quad\  P_{\fff R} \leq P_{ max },\label{sin_eta_e}  
 \end{align}
 \end{subequations}
 where $\epsilon$ and $\delta$ denote the constraints of covertness and reliability, respectively. 
 %Note that the optimization problem optimizes the target rate $T_{s}$, transmit power $P_{\fff S}$ and $P_{\fff R}$ at $S$ and $R$, respectively.
 %We first calculate the second-order derivatives of P s and Pr with respect to 53 53b to demonstrate that the function constitutes a convex problem."

By calculating the derivative of $\xi^*$ with respect to $P_{\fff S}$ and $P_{\fff R}$, we can easily see that (\ref{sin_eta_b}) is a convex constraint. Similarly, by calculating the derivative of $T_s$ with respect to $p_{out}^{(s)}$, we know (\ref{sin_eta_c}) is also a convex constraint. Since (\ref{sin_eta_d}) and (\ref{sin_eta_e}) are convex constraints, the set of constraints (\ref{sin_eta_b})-(\ref{sin_eta_e}) is then a convex set. Regarding the objective function $\eta_s$ in (\ref{sin_eta_a}), we can prove that its second-order partial derivatives respect to $P_{\fff S}$, $P_{\fff R}$ and ${T_s}$ are all non-negative, so its Hessian matrix is positive semi-definite and the $\eta_s$ is then a convex function \cite{aguilera2009convex}. Thus, we can employ the method of Lagrange multipliers and Karush-Kuhn-Tucker (KKT) conditions to solve the optimal problem in (53). For the problem (53), the corresponding Lagrangian function and KKT conditions can be determined as (54) and (55a)-(55f),
\begin{equation}
\begin{aligned}
L(P_{\fff S}, \! P_{\fff R}, \! T_s)\! =\! & -\! (1\! -p_{out}^{(s)})T_s\! +\! K_1(1\! -\epsilon\! -\xi^*)\! +\! K_2(p_{out}^{(s)}\! -\delta)\\
&\quad +K_3(P_{\fff S}-P_{max})+K_4(P_{\fff R}-P_{max}),
\label{kkt1}
\end{aligned}
\end{equation}
\begin{subequations}
\begin{align}
& (\ref{sin_eta_b}),  (\ref{sin_eta_c}),  (\ref{sin_eta_d}),  (\ref{sin_eta_e}),\label{kkt_condition_c}\\
& K_1(1-\epsilon-\xi^*)=0,\label{kkt_condition_d}\\
&K_2(p_{out}^{(s)}-\delta)=0,\label{kkt_condition_e}\\
&K_3(P_{\fff S}-P_{max})=0,\label{kkt_condition_f}\\
&K_4(P_{\fff R}-P_{max})=0,\label{kkt_condition_g}\\
& K_1\geq0,K_2\geq0,K_3\geq0,K_4\geq0,\label{kkt_condition_h}
\end{align}
\end{subequations}
where $K_1$, $K_2$, $K_3$, and $K_4$ are the Lagrange multipliers, (\ref{kkt_condition_c}) is the primal feasibility which ensures the optimal solution satisfies all the constraints, (\ref{kkt_condition_d})-(\ref{kkt_condition_g}) are the complementary slackness which states the Lagrange multipliers are complementary to the corresponding inequality constraints at the optimal solution, (\ref{kkt_condition_h}) is the dual feasibility which ensures the Lagrange multipliers associated with the inequality constraints are non-negative. 
To solve problem (53), we first establish the following equations
\begin{subequations}
\begin{align}
& {{d{L(P_{\fff S}, P_{\fff R}, T_s)}} \over {d{P_{\fff S}}}}=0,\label{kkt_condition_a}\\
& {{d{L(P_{\fff S}, P_{\fff R}, T_s)}} \over {d{P_{\fff R}}}}=0,\\
& {{d{L(P_{\fff S}, P_{\fff R}, T_s)}} \over {d{T_s}}}=0.
\end{align}
\end{subequations}
 For equation (\ref{kkt_condition_a}), we can see from (\ref{kkt1}), (\ref{two-hop dep}), (\ref{eta_s}) that it can be further formulated as
  %We first focus on the derivation of (\ref{kkt_condition_a}), it can be determined as
%We first focus on the optimal transmit power $P_{\fff S}$ at $S$. The first-order derivative of $L(P_{\fff S}, P_{\fff R}, T_s)$ respect to $P_{\fff S}$ can be determined as
%To obtain the the stationary points of $P_{\fff S}$, we calculate the first-order derivative of $L(P_{\fff S}, P_{\fff R}, T_s)$ respect to $P_{\fff S}$ which can be determined as
 \begin{equation}
\begin{aligned}
\!\!{{d{L(\!P_{\fff S},P_{\fff R},T_m\!)}} \over {d{P_{\fff S}}}}\!\!=\!&-\!\!\frac{T_s\!+\!K_2}{4P_{\fff S}\!(\ln\rho)^2}\!\!\left[\!\mathbf{Ei}(\!-\frac{\kappa_1\mu_2}{P_{\fff R}}\!)\!-\!\mathbf{Ei}(\!-\frac{\kappa_1\mu_1}{P_{\fff R}}\!)\!\right]\!\!\!\left[\! e^{\!-\!\frac{\kappa_1\mu_1}{P_{\fff S}}}\!\!\right.\\
&\!\!\left.-e^{\!-\!\frac{\kappa_1\!\mu_2}{P_{\fff S}}}\!\right]\!\!+\!\! K_1\!e^{\!-\!\frac{\mu_2}{P_{\fff R}}}\!\!\!\left[\!\mathbf{Ei}(\!\frac{\mu_1}{P_{\fff R}}\!)\!-\!\mathbf{Ei}(\!\frac{\mu_2}{P_{\fff R}}\!)\!\right]\!\!\!\left[\!\frac{\mu_2}{P_{\fff S}^2}\!e^{-\!\frac{\mu_2}{P_{\fff S}}}\!\!\right.\\
&\!\!\!\!\left.\left[\!\mathbf{Ei}(\!\frac{\mu_1}{P_{\fff S}}\!)\!\!-\!\mathbf{Ei}(\!\frac{\mu_2}{P_{\fff S}}\!)\!\right]\!\!+\!\!\frac{1}{P_{\fff S}}\!\!\!\left[\!1\!-\!e^{-\!\frac{(1\!-\rho^2)\!\mu_1}{P_{\fff S}}}\!\!\right]\!\right]\!\!+\!\!K_3\!\!=\!\!0.
\label{D_Ps}
\end{aligned}
\end{equation}
By substituting (\ref{kkt_condition_d}) and (\ref{kkt_condition_f}) into (\ref{D_Ps}), we know that the optimal solution of (53) satisfies the following equation
%By substituting the complementary slackness in (\ref{kkt_condition_d}) and (\ref{kkt_condition_f}) into (\ref{D_Ps}), the optimal transmit power $P^*_{\fff S}$ at $S$ satisfies the following condition
 \begin{equation}
\begin{aligned}
&\frac{T_s+K_2}{4P_{\fff S}\!(\ln\rho)^2}\!\!\left[\!\mathbf{Ei}(\!-\frac{\kappa_1\mu_2}{P_{\fff R}}\!)\!-\!\mathbf{Ei}(\!-\frac{\kappa_1\mu_1}{P_{\fff R}}\!)\!\right]\!\!\!\left[\! e^{\!-\!\frac{\kappa_1\mu_1}{P_{\fff S}}}\!\!-\!e^{\!-\!\frac{\kappa_1\mu_2}{P_{\fff S}}}\!\right]\!\!-\!\!\frac{\mu_2}{P_{\fff S}^2}\!K_1\!\epsilon\\
&\quad\!\!-\!\!K_1\!e^{-\frac{\mu_2}{P_{\fff S}}}\!\!\left[\mathbf{Ei}(\!\frac{\mu_1}{P_{\fff S}}\!)-\mathbf{Ei}(\!\frac{\mu_2}{P_{\fff S}}\!)\right]\!\!\frac{1}{P_{\fff S}}\!\!\left[1-e^{-\frac{(1-\rho^2)\mu_1}{P_{\fff S}}}\right]\!\!-\!\!K_3\!=\!0.
\label{kkt_ps}
\end{aligned}
\end{equation}
Following a similar analysis as above, we can see that the optimal solution of (53) also satisfy the following equations
%the optimal transit power $P^*_{\fff R}$ at $R$ and the optimal target rate $T_s^*$ satisfy the following conditions, respectively
 \begin{equation}
\begin{aligned}
&\frac{T_s+K_2}{4P_{\fff R}\!(\ln\rho)^2}\!\!\left[\!\mathbf{Ei}(\!-\frac{\kappa_1\!\mu_2}{P_{\fff S}}\!)\!\!-\!\mathbf{Ei}(\!-\frac{\kappa_1\!\mu_1}{P_{\fff S}}\!)\!\right]\!\!\!\left[\! e^{\!-\!\frac{\kappa_1\!\mu_1}{P_{\fff R}}}\!\!-\!e^{\!-\!\frac{\kappa_1\!\mu_2}{P_{\fff R}}}\!\right]\!\!-\!\!\frac{\mu_2}{P_{\fff R}^2}\!K_1\!\epsilon\\
&\quad\!\!-\!\!K_1\!e^{-\frac{\mu_2}{P_{\fff R}}}\!\!\left[\mathbf{Ei}(\!\frac{\mu_1}{P_{\fff R}}\!)\!-\!\mathbf{Ei}(\!\frac{\mu_2}{P_{\fff R}}\!)\right]\!\!\frac{1}{P_{\fff R}}\!\!\left[1\!-\!e^{-\!\frac{(\!1\!-\rho^2\!)\mu_1}{P_{\fff R}}}\right]\!\!-\!\!K_4\!=\!0,
\label{kkt_pr}
\end{aligned}
\end{equation}
\begin{equation}
\begin{aligned}
&T_s(1-\delta)+\frac{T_s-K_2}{4(\ln\rho)^2}\left[A_1\left[\mathbf{Ei}(-\frac{\kappa_1\mu_2}{P_{\fff R}})-\mathbf{Ei}(-\frac{\kappa_1\mu_1}{P_{\fff R}})\right]\right.\\
&\left.\quad+A_2\left[\mathbf{Ei}(-\frac{\kappa_1\mu_2}{P_{\fff S}})-\mathbf{Ei}(-\frac{\kappa_1\mu_1}{P_{\fff S}})\right]\right]=0,
\label{kkt_ts}
\end{aligned}
\end{equation}
where $A_1=\frac{2^{2T_s+1}\log2}{\kappa_1}e^{-\frac{\kappa_1\mu_2}{P_{\fff S}}}-e^{-\frac{\kappa_1\mu_1}{P_{\fff S}}}$, $A_2=\frac{2^{2T_s+1}\log2}{\kappa_1}e^{-\frac{\kappa_1\mu_2}{P_{\fff R}}}-e^{-\frac{\kappa_1\mu_1}{P_{\fff R}}}$. Based on equations (\ref{kkt_ps})-(\ref{kkt_ts}), we can then get the optimal solution ($P_{\fff S}$, $P_{\fff R}$, $T_s$) of (53). By substituting ($P_{\fff S}$, $P_{\fff R}$, $T_s$) into (\ref{d_eta}), we can obtain the maximum covert throughput $\eta_s^*$ in the single antenna scenario.
%Based on functions (\ref{kkt_ps})-(\ref{kkt_ts}), the global optimal solution of $P_{\fff S}$, $P_{\fff R}$, $T_s$ which maximizes the covert throughput $\eta_s^*$ can be obtained. Thus, the maximum covert throughput is determined as 

\subsection{Covert Throughput Maximization in Multi-antenna Scenario}
 From \textbf{Theorem 3} we can see that the covert throughput in the multi-antenna scenario is related to $P_{\fff S}$, $P_{\fff R}$ and ${T_m}$. Here, we determine the optimal settings of these parameters to achieve the maximum covert throughput $\eta_m^*$ in this scenario. The related optimization problem can be formulated as follows.
 \begin{subequations}
 \begin{align}
&\max_{P_{\fff S}, P_{\fff R}, T_{m}}\eta_{m}  \label{mul_eta_a}\\
& \quad s.t. \quad \xi ^*\geq 1- \epsilon ,   \label{mul_eta_b}  \\
 &\!\!\!\qquad\qquad  p_{out}^{(m)} \leq  \delta ,  \label{mul_eta_c}  \\
&\qquad\quad\  P_{\fff S} \leq P_{ max },\label{mul_eta_d}  \\
&\qquad\quad\  P_{\fff R} \leq P_{ max }.\label{mul_eta_e}  
 \end{align}
 \end{subequations}
 
Since the objective function in (\ref{mul_eta_a}) and the reliability constraint in (\ref{mul_eta_c}) are too complex to prove their convexity (if it is not possible), we apply the classic direct search algorithm in \cite{hooke1961direct} to devise the following Algorithm 1 for solving the optimization problem in (61).

\begin{algorithm}
    \caption{Covert throughput maximization in multi-antenna scenario}
    \label{alg:2} 
    \begin{algorithmic}[1] 
     \REQUIRE Maximum transmit power constraint $P_{max}$, Covertness constraint $ \epsilon$, Reliability constraint $ \delta $;
        \ENSURE Maximum covert throughput $ \eta_{m}^*$;
        \STATE Initialize $P_{\fff S}=0$, $P_{\fff R}=0$, $T_m=0$;
        \STATE Set the step size $h_1$ for $P_{\fff S}$, $h_2$ for $P_{\fff R}$, and $h_3$ for $T_m$; the convergence tolerance $\phi>0$; the iteration index $v=0$, and the maximum number of iterations $v_{max}$;
         \WHILE{$P_{\fff S} \leq P_{max}$}
           \WHILE{$P_{\fff R} \leq P_{max}$}
           %\WHILE{$\mid \eta_m^*-\eta_m^{max}\mid\geq\phi$}
          \STATE  Calculate $\xi^*$ according to (\ref{two-hop dep});
          \IF{$\xi^* \geq 1-\epsilon$}
          \STATE According to $p_{out}^{(m,1)}$ and $p_{out}^{(m,2)}$, update the value of $p_{out}^{(m)}$; 
          %\STATE  Calculate $p_{out}^{(m,1)}$ according to (\ref{pout_mul_1});
           \WHILE{$p_{out}^{(m)}\leq \delta$} 
           %\STATE  Calculate $p_{out}^{(m)}$ according to (\ref{pout_m});
       %\STATE Attempt one-dimensional search in the $T_m$ direction; 
       \STATE $v=v+1$;
       \STATE According  to (\ref{mul_eta_a}), search the maximum value $\eta_m^{max}$ of $\eta_m$; 
       \STATE \textbf{break} \textbf{if} $\mid \eta_m^*-\eta_m^{max}\mid \leq \phi$ or $v \geq v_{max}$;
       % \STATE Find out the $\eta_m^*$ which makes the $\eta_m(P_{\fff S}, P_{\fff R}, T_{m}+h)$ in (\ref{mul_eta_a}) obtain the maximum value;
       %\STATE Find out the $\eta_m^*$ which makes the $\eta_m$ in (\ref{mul_eta_a}) obtain the maximum value
      % \STATE Calculate $\eta_m$ according to Theorem 3;
        %\STATE Update $\eta_m^*$ to the maximum $\eta_m$;  
        \IF{$\eta_m^{max}>\eta_m^*$}
        \STATE $\eta_m^*=\eta_m^{max}$;
        \ENDIF
         \STATE $T_m=T_m+h_3$;
         %\STATE $T_m$ = $T_m+h$;
         %\STATE $v=v+1$;
         %\STATE Update $T_m$ to the next value;
         \ENDWHILE
           \ENDIF
          %\STATE Update $P_{\fff R}$ to the next value;
         \STATE $P_{\fff R}=P_{\fff R}+h_2$;
         \ENDWHILE
         %\STATE Update $P_{\fff S}$ to the next value;
         \STATE $P_{\fff S}=P_{\fff S}+h_1$;
         \ENDWHILE
     \RETURN $\eta_{m}^*$;
    \end{algorithmic}
\end{algorithm}

 \section{Numerical Results}
 In this section, we provide extensive numerical results to validate our theoretical models as well as to illustrate the impacts of system parameters on system performance and the covert communication performance enhancement from adopting multi-antenna technique. For simplicity, in the following numerical results, we set $N_{\fff S}=N_{\fff R_{t}}=N_t$ and $N_{\fff R_{r}}=N_{\fff D}=N_r$ in the multi-antenna scenario.
% In this section, we provide extensive numerical results to explore the covert communication performance enhancement from adopting multi-antenna technique in a two-hop relay system. We also show the impact of system parameters on the covert performance, i.e., transmit power, noise uncertainty, background noise, and covert constraints. Unless otherwise stated, in the numerical results, the related parameters are set as $P_{\fff S}=5W$, $P_{\fff R}=5W$, $\sigma_n^2=-5dB$, and $\rho=1.5$.
 %Here, we set the $P$ as transmit power both in single-antenna/multi-antenna scenarios. Likewise, $N_t$ represents the transmit antenna and $N_r$ represents the receive antenna.
 %\subsection{Model Validation}
 
 \subsection{DEP Performance}
To validate our theoretical model on minimum DEP $\xi^*$, we conducted related simulations and summarize in Fig. \ref{p_dep} both the simulation results and theoretical ones for $\xi^*$ vs. $P_{\fff S}=P_{\fff R}=P$ under the settings of $\rho=1.5$, $\sigma_n^2=\{-10,-5,3\}$dBm. Fig. \ref{p_dep} shows the theoretical results match nicely with the simulated ones, indicating the theoretical model in (\ref{two-hop dep}) can accurately depict the DEP performance in the concerned system. Fig. \ref{p_dep} also shows that when $P_{\fff S}=P_{\fff R}=P$, $\xi^*$ decreases monotonously as $P$ increases. This is because a larger $P$ makes it easier for warden to detect the transmission process and leads to a smaller DEP. We can also see from Fig. \ref{p_dep} that for a given setting of $P$, $\xi^*$ increases as the nominal noise  $\sigma_n^2$ increases. This is mainly due to the reason that a larger $\sigma_n^2$ will make it harder for warden to detect the presence of the signal.
%To validate our theoretical models, we first conduct simulations to evaluate DEP. We plot in Fig. \ref{p_dep} both the simulation results and theoretical ones for minimum DEP $\xi^*$ vs. transmit power $P_{\fff S}=P_{\fff R}=P$ with the settings of $\rho=1.5$, $\sigma_n^2=\{-10,-5,3\}$dB. From Fig. \ref{p_dep} we can see that the theoretical $\xi^*$ matches nicely with the simulated one, indicating that our theoretical models can accurately depict the DEP performance, which validates Theorem 1. We can see from Fig. \ref{p_dep}  that $\xi^*$ decreases monotonously as $P$ increases. This is because a larger $P$ makes the transmission process easier to be detected by the warden, and thus a smaller DEP. We can also see from Fig. \ref{p_dep} that for a given setting of $P$, $\xi^*$ increases as the nominal noise $\sigma_n^2$ increases. This is mainly due to the reason that a smaller $\sigma_n^2$ will have less effect on the warden, and thus the detector can more easily detect the presence of the signal. 
 \begin{figure*}
  \begin{minipage}[b]{.3\textwidth}
    \centering
    \includegraphics[width=6.5cm]{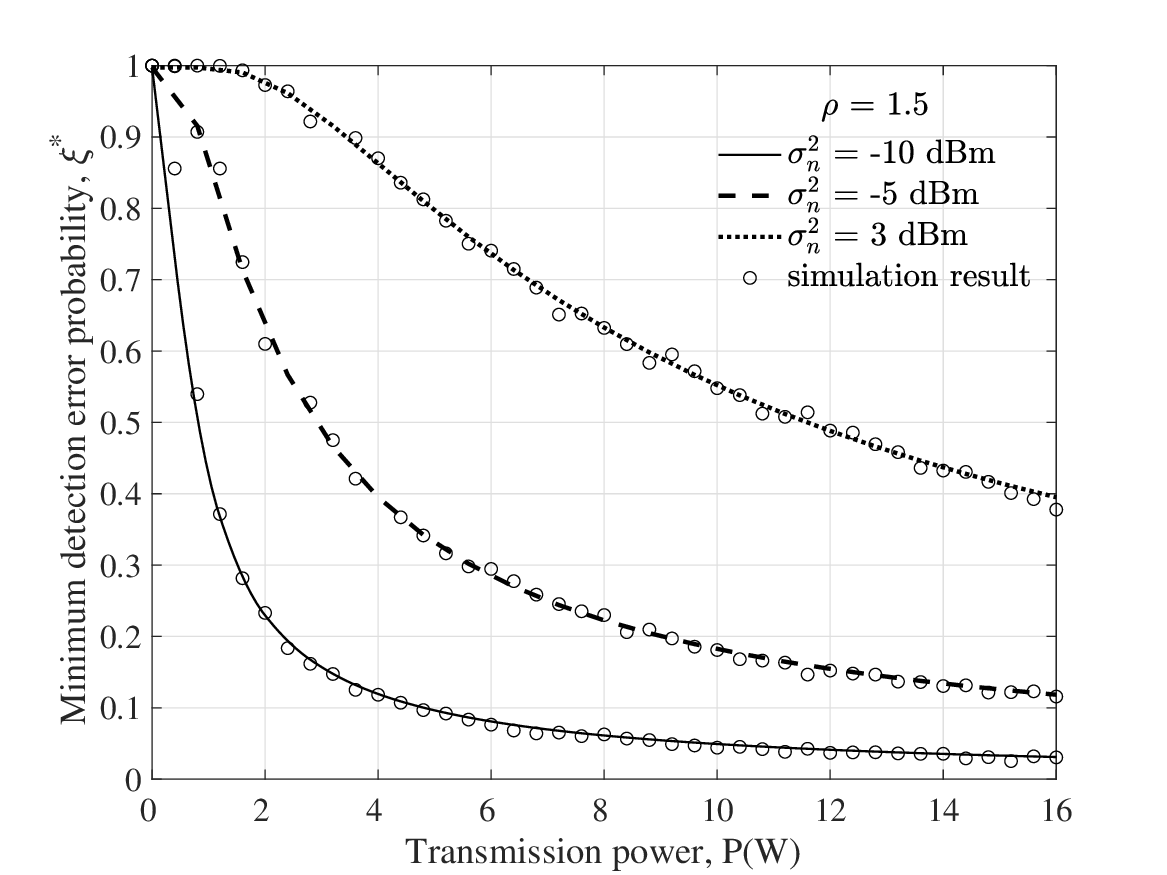}
    \caption{Minimum DEP $\xi^*$ vs. transmission power $P$.}\label{p_dep}
  \end{minipage}\hfill
  \begin{minipage}[b]{.3\textwidth}
    \centering
    \includegraphics[width=6.5cm]{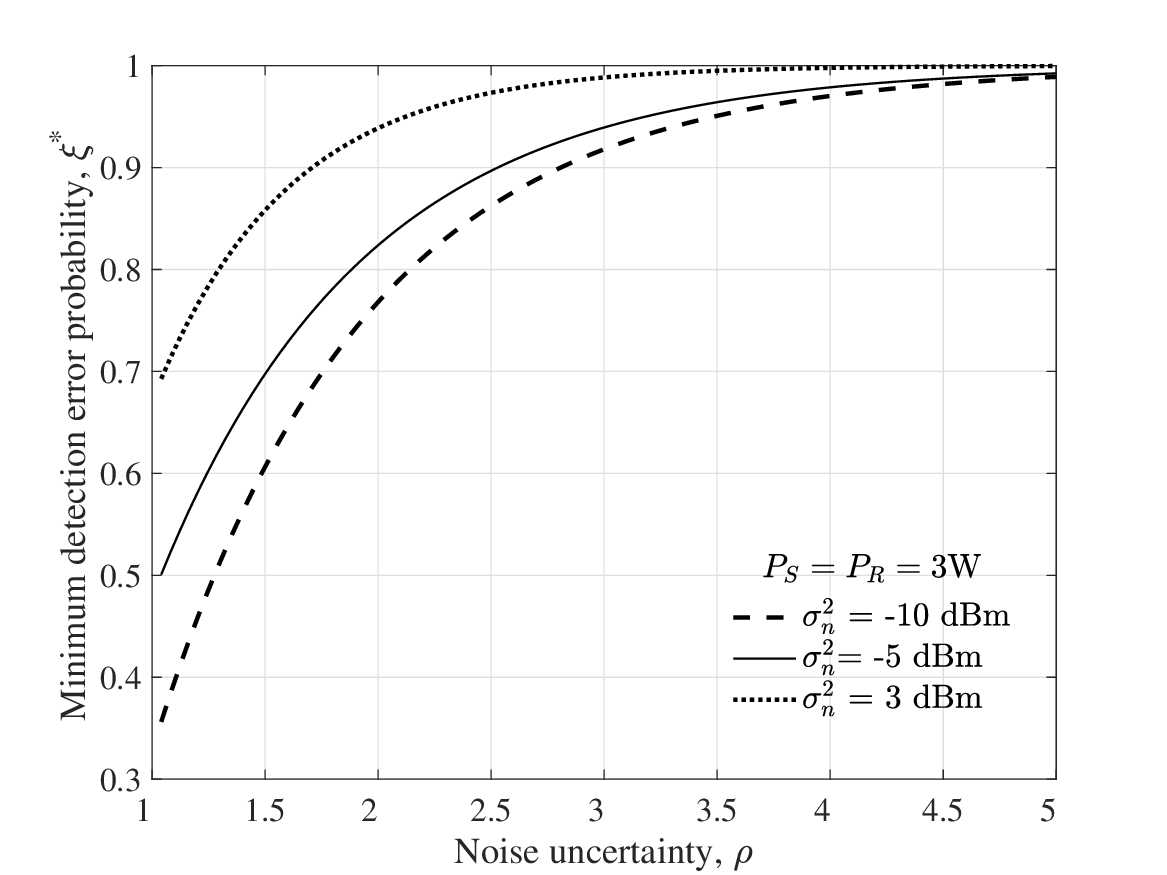}
    \caption{Minimum DEP $\xi^*$ vs. noise uncertainty $\rho$.}\label{n_dep}
  \end{minipage}\hfill
  \begin{minipage}[b]{.3\textwidth}
    \centering
    \includegraphics[width=6.5cm]{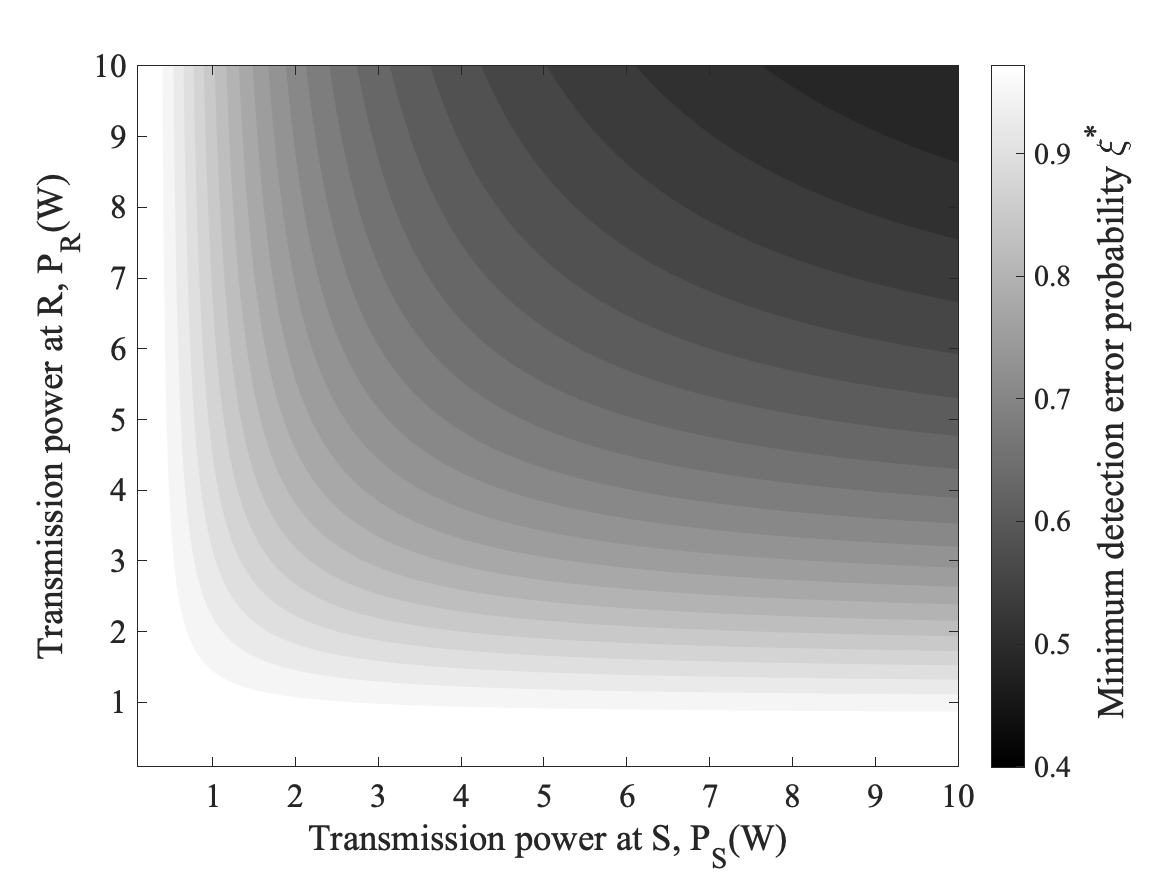}
    \caption{Minimum DEP $\xi^*$ vs. $P_{\fff S}$ and $P_{\fff R}$. $\sigma_n^2=-5$dBm, $\rho = 1.5$.}\label{compare}
  \end{minipage}
\end{figure*}

To explore the impact of noise uncertainty $ \rho $ on $ \xi^*$, we summarize in Fig. \ref{n_dep} $ \xi^*$ vs. $ \rho $ with the settings of $P_{\fff S}=P_{\fff R}=3$W, $\sigma_n^2=\{-10, -5, 3\}$dBm. As shown in Fig. \ref{n_dep} that $ \xi^*$ increases monotonously as $ \rho $ increases. This is because when $P_{\fff S}$, $P_{\fff R}$, and $\sigma_n^2$ are fixed, increasing $ \rho$ will be helpful for hiding the existence of covert communication and thus for degrading the detection performance of the warden. Fig. \ref{n_dep} also shows that even as $\rho$ tends to 1, $\xi^*$ will not reduce zero. This is because although no noise uncertainty will diminish as $\rho$ tends 1, the nominal noise can still help to hide the signal transmission. The above observation indicates that the noise uncertainty and nominal noise can jointly affect the warden's detection performance.
% that increasing $\rho$ can reduce the detection performance of the warden when the $P_{\fff S}$, $P_{\fff R}$, and $\sigma_n^2$ are given, and a high level of noise is better for hiding the existence of communication, such that it can improve the covert performance of the system. We can also see from Fig. \ref{n_dep} that $\xi^*$ is not equal to zero as $\rho$ tends to 1. This is because the background noise in (\ref{pdf-nw}) follows log-uniform distribution. When $\rho$ is close to 1, it means there is no noise uncertainty, the nominal noise can still help to hide the signal transmission, thus achieving covert communication. The above observation indicates that the noise uncertainty and nominal noise closely affect the warden's detection performance.
%how $\xi^*$ varies with the increasing of $ \rho $ for a

%We then investigate the impact of noise uncertainty $ \rho $ on the minimum DEP $ \xi^*$ as shown in Fig. \ref{n_dep}. %We also can observe that, for a given noise uncertainty, increasing the background noise $ \sigma_n ^2$ can also improve  $ \xi^*$, and determine the upper bound of the detection error probability. It indicates that with the increase of the background noise, the minimum detection error probability of the system can be improved. This is because a larger background noise can better become the signal's shelter, which is conducive to improving covert performance. It can be seen from Fig. \ref{n_dep} that the noise uncertainty and background noise closely affect the warden's detection performance.

We further show in Fig. \ref{compare} how $\xi^*$ varies when $P_{\fff S}$ and $P_{\fff R}$ change independently under the settings of $\sigma_n^2=-5$dBm and $\rho=1.5$. As dipicted in Fig. \ref{compare} a smaller $\xi^*$ can be achieved only when both $P_{\fff S}$ and $P_{\fff R}$ becomes larger, while a smaller $P_{\fff S}$ and/or $P_{\fff R}$ will always lead to a larger $\xi^*$. This is because we can see from (\ref{2-dep}) that in the concerned two-hop relay system, warden can successfully detect the end to end covert communication only when he detects the covert transmissions in both hops.

\subsection{Covert Throughput Performance}

\begin{comment} 
\begin{figure}[h]
\centering 
%\includegraphics[width=0.5\textwidth]{fig0.jpg}
\includegraphics[width=9cm]{compare_2.eps}
\caption{Covert throughput $\eta$ vs. target rate $T$.}\label{compare_2}
\end{figure}

\begin{figure}[t!]
\centering 
%\includegraphics[width=0.5\textwidth]{fig0.jpg}
\includegraphics[width=9cm]{sig.eps}
\caption{Impact of reliability and covertness constraints on covert throughput in multi-antenna scenario.}\label{constraints}
\end{figure}

\begin{figure}[t!]
\centering 
%\includegraphics[width=0.5\textwidth]{fig0.jpg}
\includegraphics[width=9.2cm, height=6.7cm]{mul.eps}
\caption{Impact of transmit power on covert throughput in multi-antenna scenario.}\label{p_ct}
\end{figure}
\end{comment} 
To verify our theoretical models on covert throughput, we show in Fig. \ref{compare_2} the simulation results and theoretical ones on $\eta$ (i.e., $\eta_s$ and $\eta_m$) vs. target rate $T$ under both single and multi-antenna scenarios, where $\epsilon = 0.15$, $N_t=2$, $N_r=8$. The results in Fig. \ref{compare_2} show that the theoretical results agree with the corresponding simulation results, so our theoretical models in (\ref{eta_s}) and (\ref{mul-eta}) can efficiently capture the behaviors of the covert throughput performance in the concerned system. We can observe from Fig. \ref{compare_2}  that as $T$ increases, both $\eta_s$ and $\eta_m$ first increase and then decrease. This is due to the reason that the effects of $T$ on $\eta$ are two-fold. On the one hand, as $T$ increases, more information will be transmitted per unit time, leading to a higher $\eta$. On the other hand, as $T$ increases further, the receiver can not conduct decoding timely, resulting in more transmission outage events and thus a decrease of $\eta$. The Fig. \ref{compare_2} also indicates clearly that an optimal setting of $T$ exists to achieve the maximum covert throughout in either single or multi-antenna scenarios.

\begin{figure}[t!]
\centering 
\includegraphics[width=0.9\linewidth]{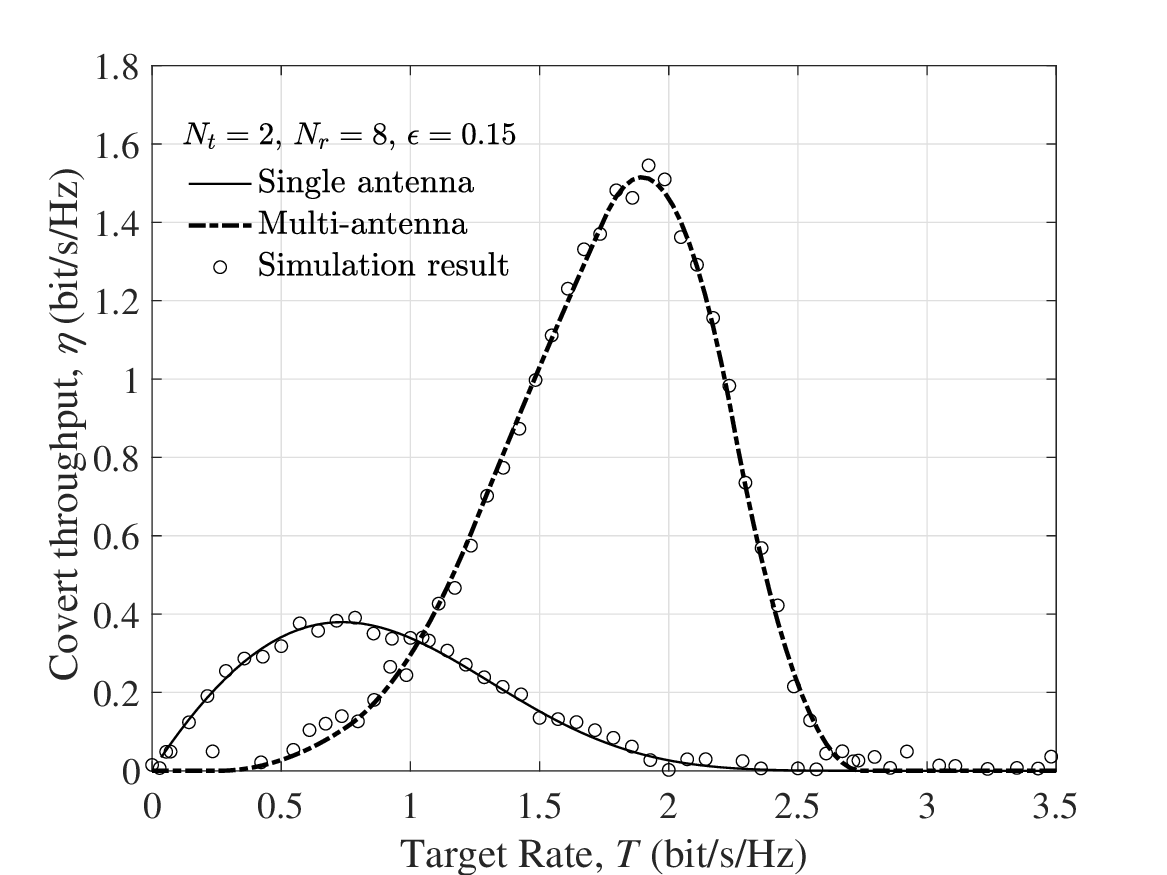}
		\caption{Covert throughput $\eta$ vs. target rate $T$.}\label{compare_2}
\end{figure}

\begin{figure}[t!]
\centering 
\includegraphics[width=9cm]{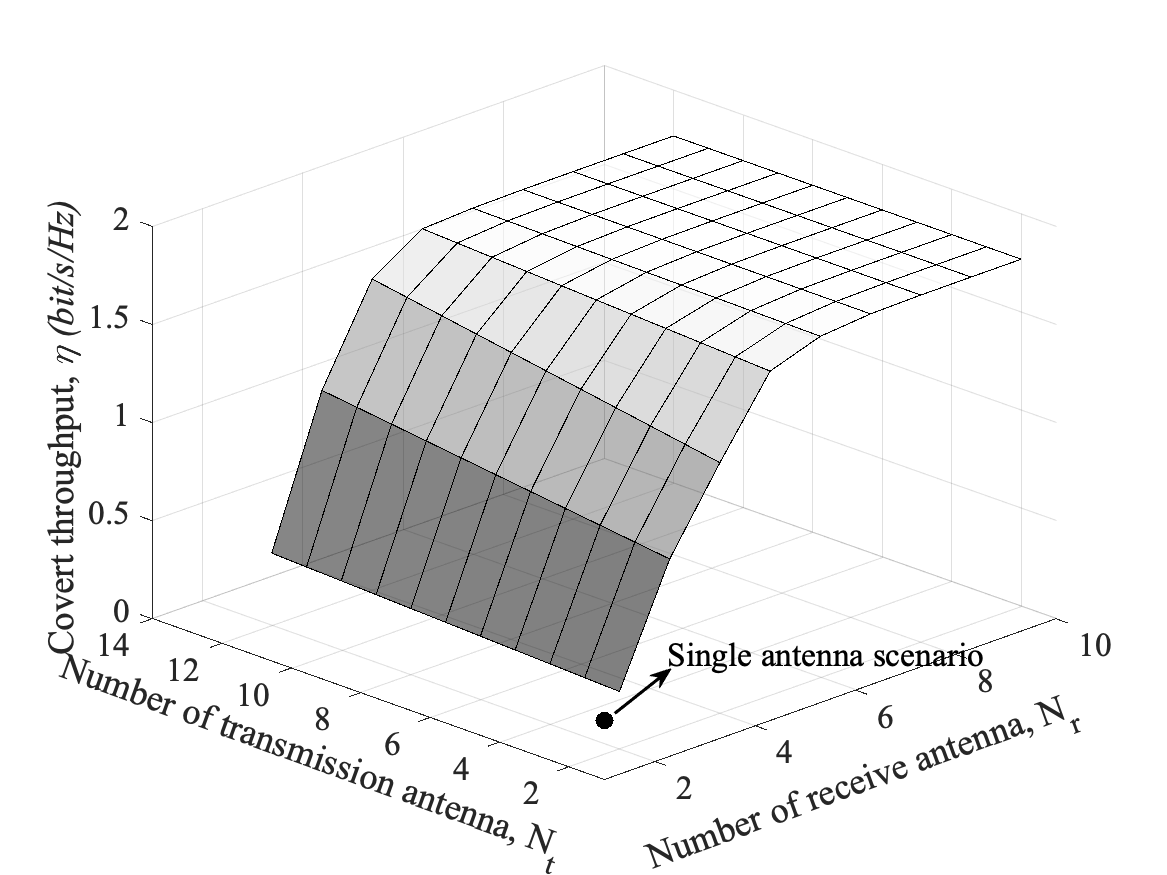}
\caption{Covert throughput $\eta$ vs. number of antennas, $(N_t, N_r)$. $\rho = 1.5$, $\sigma_n^2=-5$dBm, $T_s = T_m = 1.5$bit/s/Hz.}\label{ante}
\end{figure}
To demonstrate the covert throughput enhancement from adopting multiple antennas, we show in Fig. \ref{ante} how $\eta_m$ varies with ($N_t$, $N_r$) under the settings of $\rho = 1.5$, $\sigma_n^2=-5$dBm, $T_s = T_m = 1.5$bit/s/Hz. We can see from Fig. \ref{ante} that compared with the single antenna scenario (i.e., when $N_t=N_r=1$), although increasing $N_t$ and/or $N_r$ can in general lead to an improvement in covert throughput, such improvement is more sensitive to the variation of $N_r$ rather than $N_t$. The reason for this phenomenon can be summarized as following. First, in the concerned system, a transmitter (either $S$ or $R$) adopts the TAS scheme to choose only the antenna with the largest SNR as the transmitting antenna, so adopting more antennas at the transmitter has limited effect on the improvement of covert throughput. Second, a receiver (either $R$ or $D$) adopts the MRC scheme to yield the maximum SNR ratio of the received signal by weighting the signals from all receive antennas, so employing more antennas at the receiver will result in a larger SNR ratio during signal transmission and leads to a higher covert throughput. However, due to the limitations of total bandwidth and transmit power, the covert throughput will no longer increases and remains constant if we increase $N_r$ further beyond a threshold.

%To demonstrate the impact of the number of transmission antenna $N_t$ and receive antenna $N_r$ on $\eta_m$, we summarize in Fig. \ref{ante} how $N_t$ and $N_r$ varies with $\eta_m$ under the settings of $\rho = 1.5$, $\sigma_n^2=-5$dB, $T_s = T_m = 1.5$bit/s/Hz. We can see from Fig. \ref{ante} that $\eta_{m}$ slightly increases as $N_t$ increases. This is because the transmitter adopts the TAS scheme which chooses the largest SNR of the antenna from transmitter to receiver. Thus, more $N_t$ has a limited effect on the $\eta_{m}$ of the system. We can also see from Fig. \ref{ante} that $\eta_{m}$ increases and then almost remains unchanged as $N_r$ increases. This phenomenon can be explained by the fact that the receiver adopts the MRC scheme which weights the signals at each receive antenna to yield the maximum SNR ratio of the received signal. More $N_r$ results in a larger SNR during signal transmission and leads to a higher $\eta_{m}$. Since the bandwidth and transmit power are limited, the $\eta_{m}$ remains unchanged when $N_r$ increases further.\par
%Since the transmit power is limited by transmit power constraint Pmax, the ηs keeps unchanged when ε increases further. 
We then illustrate in Fig. \ref{compare_3} and Fig. \ref{mum_op} how covert throughput varies with the covertness constraint and noise uncertainty. The results in these two figures indicate clearly that the covert throughput can be significantly improved by employing multiple antennas, and such improvement can be further enhanced by applying the covert throughput maximization. For example, as depicted in Fig. \ref{compare_3} that when $\epsilon = 0.15$, the covert throughput in the single-antenna scenario is $\eta_s = 0.355 $bit/s/Hz, while the corresponding maximal covert throughput is $\eta_s^* = 0.387 $bit/s/Hz, and the covert throughput then improves to $\eta_m = 1.359 $bit/s/Hz by employing multiple antennas of ($N_t$ = 2, $N_r$ = 8) and further to $\eta_m^* = 1.563$bit/s/Hz with the help of covert throughput maximization. Fig. \ref{compare_3} shows that as $\epsilon$ increases (i.e., as covertness constraint becomes less stringent), $\eta$ first increase and then keep constant. This is because as $\epsilon$ increases, the transmitter can adopt a larger power to transmit the signals, leading to a larger $\eta$. However, the transmit power is limited by the transmit power constraint $P_{max}$, so $\eta$ keep unchanged when $\epsilon$ increases further. The results in Fig. \ref{mum_op} indicate that as $\rho$ increases (i.e., as the noise uncertainty becomes more intensive), $\eta$ first increase and then decrease. This can be explained as follows. We can see from (\ref{two-hop dep}), (\ref{pout_40}) and (\ref{pout_m}) that $\rho$ is related to both the DEP and outage probability, so the effect of $\rho$ on $\eta$ are two folds. Although increasing $\rho$ is helpful for hiding the covert communication and thus for increasing $\eta$, but a too large $\rho$ will result in a high outage probability and a decrease of $\eta$. Therefore, noise uncertainty should be properly controlled to achieve the optimal covert throughput performance.

\begin{figure}[t!]
\centering 
\includegraphics[width=9cm]{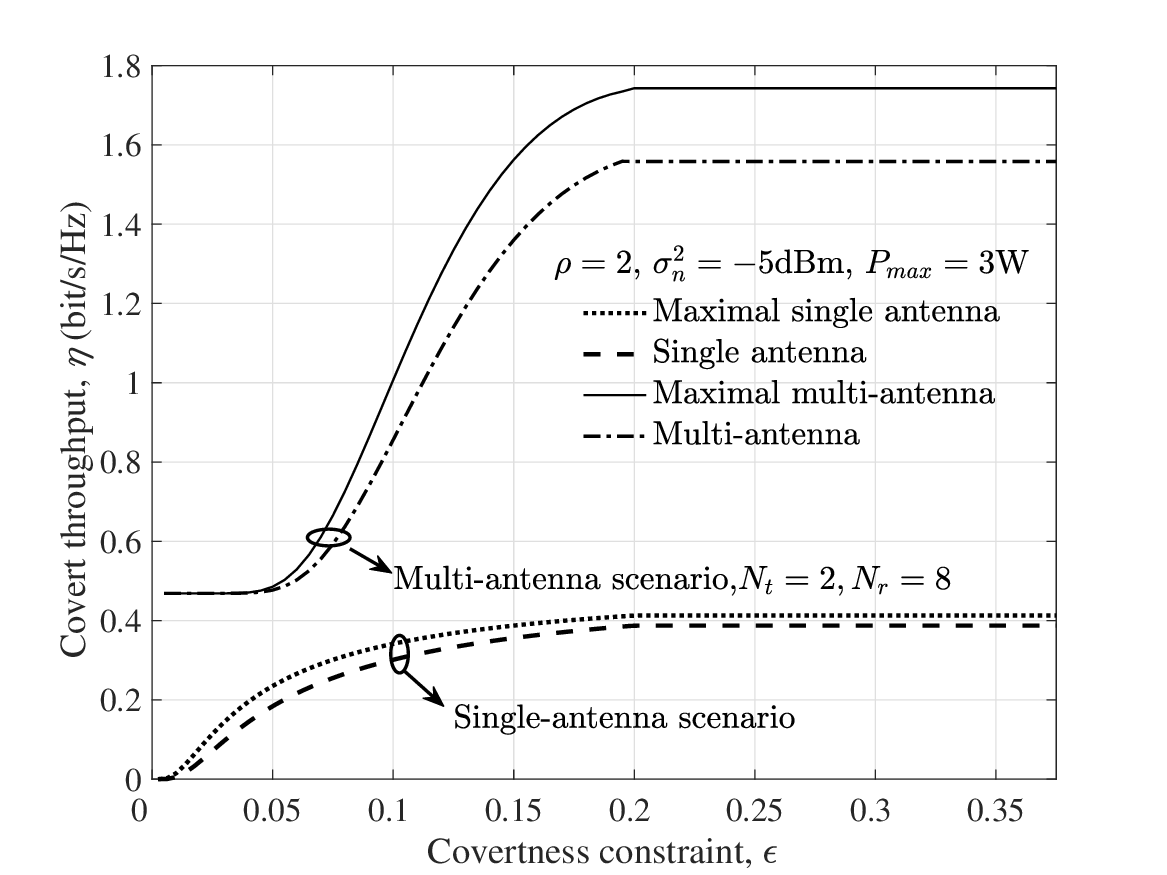}
\caption{Covert throughput $\eta$ vs. covertness constraint $\epsilon$.}\label{compare_3}
\end{figure}

\begin{figure}[t]
\centering 
\includegraphics[width=9cm, height=6.5cm]{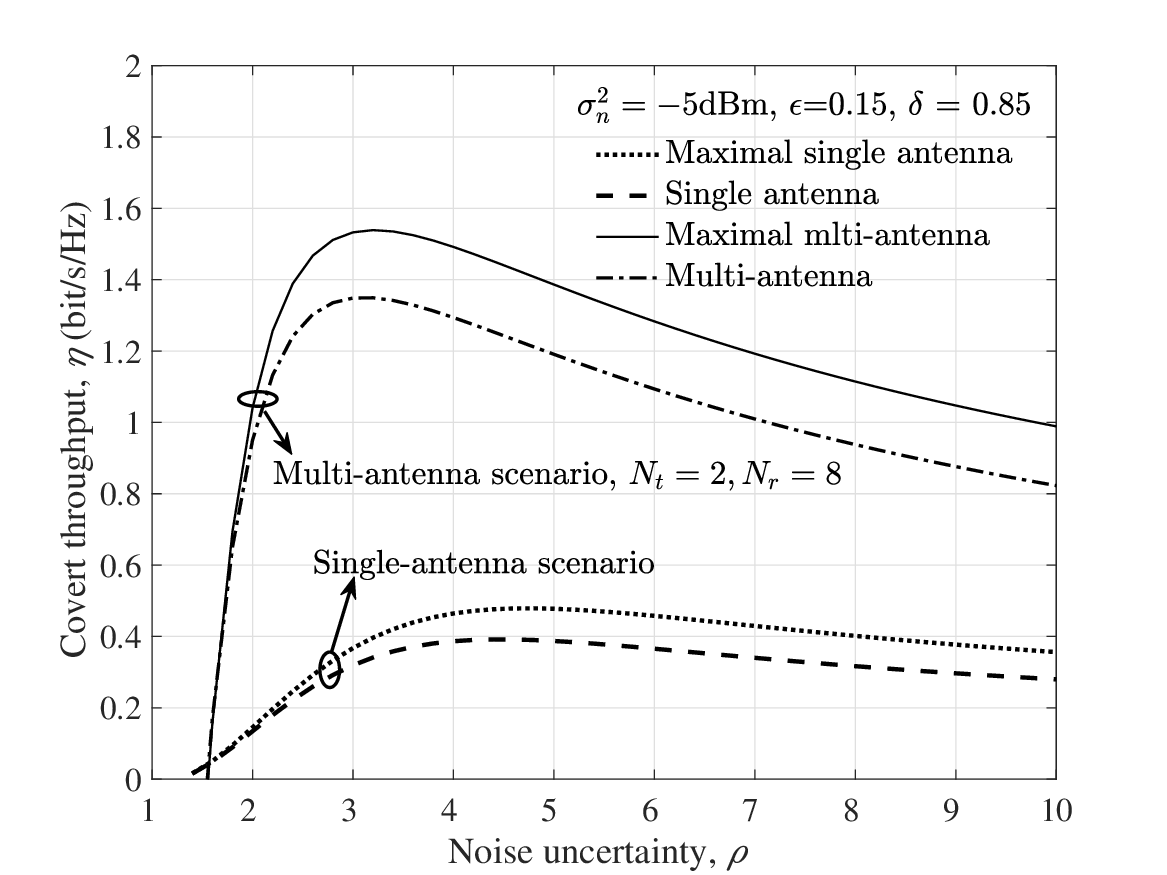}
\caption{Covert throughput $\eta$ vs. noise uncertainty $\rho$.}\label{mum_op}
\end{figure}

\section{Conclusion}
This paper developed related theoretical frameworks to reveal the impact of employing multi-antenna technique on covert communication performance in a two-hop relay system. Our results indicate that adopting multiple antennas has a significant impact on covert performance in relay systems. In particular, even under the strict constraints of covertness, reliability and transmit power, in comparison with the simple single-antenna scenario, applying the multi-antenna technique can still lead a great improvement in covert throughput, and such improvement in general increases as the number of antennas increases. It is expected that this work can shed light on the covert communication enhancement in wireless relay systems.

\vspace*{1\baselineskip}
%\begin{thebibliography}{1}
\bibliographystyle{IEEEtran}
\bibliography{refe.bib}

\end{document}